%% file: journal.tex
\documentclass[10pt]{svproc}
\let\HH\H

\newcommand{\Konig}{K\HH onig}

\usepackage[T1]{fontenc}
\usepackage[utf8]{inputenc}
\usepackage[english]{babel}
\usepackage{enumerate}%
\usepackage{bbm}
\usepackage{mathpartir}

\DeclareUnicodeCharacter{22A2}{$\vdash$}
\include{macros}

\usepackage{hyperref}
\usepackage{cleveref}
\usepackage{graphicx}

\crefname{lemmaAux}{Lemma}{Lemmas}
\crefname{theoremAux}{Theorem}{Theorems}
\crefname{definitionAux}{Definition}{Definitions}
\crefname{factAux}{Fact}{Facts}
\crefname{corollaryAux}{Corollary}{Corollarys}

\usepackage[countsame]{coqtheorem}
\setBaseUrl{{http://www.ps.uni-saarland.de/extras/fol-completeness-ext/website/Undecidability.FOLC.}}

\renewcommand\Form{\MBB{F}}

\usepackage{todonotes}

\title{Completeness Theorems for First-Order Logic Analysed in Constructive Type Theory}
\subtitle{Extended Version}
\titlerunning{Completeness Theorems for FOL Analysed in Constructive Type Theory}
\author{Yannick Forster\inst{1} \href{https://orcid.org/0000-0002-8676-9819}{\protect\includegraphics[scale=.25]{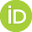}}, Dominik Kirst\inst{1} \href{https://orcid.org/0000-0003-4126-6975}{\protect\includegraphics[scale=.25]{orcid.png}}, Dominik Wehr\inst{1,2} \href{https://orcid.org/0000-0001-6456-8111}{\protect\includegraphics[scale=.25]{orcid.png}}}
\authorrunning{Yannick Forster, Dominik Kirst, Dominik Wehr}
\institute{{Saarland University, Saarland Informatics Campus \\Saarbrücken, Germany}\\
  \path|{forster,kirst}@ps.uni-saarland.de|
  \and
  {Institute for Logic, Language and Computation, University of Amsterdam \\
      Amsterdam, The Netherlands} \\
  \path|dwehr@dortselb.st|}

\newenvironment{proofqed}{\begin{proof}}{\qed\end{proof}}
\renewcommand\L{\mathsf{L}}

\newcommand{\WKL}{\ensuremath{\mathsf{WKL}}}

\begin{document}
\maketitle
\setcounter{footnote}{0}

\begin{abstract}
	We study various formulations of the completeness of first-order logic phrased in constructive type theory and mechanised in the Coq proof assistant.
	Specifically, we examine the completeness of variants of classical and intuitionistic natural deduction and sequent calculi with respect to model-theoretic, algebraic, and game-theoretic semantics.
	As completeness with respect to the standard model-theoretic semantics \`a la Tarski and Kripke is not readily constructive, we analyse connections of completeness theorems to Markov's Principle and Weak \Konig's Lemma and discuss non-standard semantics admitting assumption-free completeness.
	We contribute a reusable Coq library for first-order logic containing all results covered in this paper.
	
\end{abstract}

\section{Introduction}

\enlargethispage{3\baselineskip}
 Completeness theorems are central to the field of mathematical logic.
 Once completeness of a sound deduction system with respect to a semantic account of the syntax is established, the infinitary notion of semantic validity is reduced to the  algorithmically tractable notion of syntactic deduction.
 In the case of first-order logic, being the formalism underlying traditional mathematics based on a set-theoretic foundation, completeness enables the use of semantic techniques to study the deductive consequence of axiomatic systems.
 
 The seminal completeness theorem for first-order logic proven by Gödel~\cite{godel_vollstandigkeit_1930} and later refined by Henkin~\cite{henkin_completeness_1949,hasenjaeger_bemerkung_1953} guarantees the existence of a syntactic deduction of every formula valid in the canonical Tarski semantics, which is based on interpreting the function and relation symbols in models providing the corresponding structure.
However, this result may not be understood as an effective procedure in the sense that a formal deduction for a formula satisfied by all models can be computed by an algorithm, since even for finite signatures the proof relies on non-constructive assumptions.
It was already known to Gödel that for a completeness proof the classically vacuous but constructively contested\footnote{\scriptsize Accepted in Russian constructivism while in conflict with Brouwer's intuitionism}
assumption of \emph{Markov's Principle}, asserting that every non-diverging computation terminates, is necessary~\cite{KreiselMP}.
Moreover, Gödel implicitly used a choice principle known as \emph{Weak \Konig's Lemma}~\cite{konig1927schlussweise} and it is a well-known result of reverse mathematics that, over classical logic, the completeness theorem is in fact equivalent to Weak \Konig's Lemma~\cite{simpson2009subsystems}.

The aim of this paper is to coherently analyse the assumptions necessary to prove completeness theorems concerning various semantics and deduction systems.
For the analysis to be as precise as possible, we choose constructive type theory with an impredicative (and thus separate) universe of propositions as base system, a formalisation of intuitionistic logic with virtually no choice principles provable without assumptions.
Concretely, we work in the \emph{polymorphic calculus of cumulative inductive constructions}~(\textsf{pCuIC})~\cite{sozeau:hal-02167423} underlying the Coq proof assistant~\cite{Coq} and in fact all results in this paper are mechanised in Coq, yielding Coq programs for constructively given completeness proofs.
For ease of language, we reserve the term ``constructive'' for statements provable in this specific system, hence in particular Markov's Principle is classified as non-constructive~\cite{coquand_independence_2017,PedrotMP}.

Coming with an internal notion of computation, constructive type theory allows us to state Markov's Principle both internally as
$$\MP:=\forall f:\Nat\to\Bool.\,\neg\neg(\exists n.\,f\,n=\btrue)\to \exists n.\,f\,n=\btrue$$
and similarly for any concrete model of computation ($\MPL$), whereby the former implies the latter. %
The second principle involved in Gödel's proof, Weak \Konig's Lemma ($\WKL$), is a function existence principle asserting that every infinite binary tree has an infinite path.
$\WKL$ is not constructive, because it is equivalent to a combination of a weak classical logical axiom and a weak choice axiom~\cite{berger2012weak}, both deemed independent in \textsf{pCuIC}.
The two main questions in focus are which of these assumptions are necessary for particular formulations of completeness and how the statements can be modified such that they hold constructively.

Applying this agenda to Tarski semantics, a first observation is that the model existence theorem, central to Henkin's completeness proof, holds constructively~\cite{HerbelinHenkin} for the $\to,\forall,\bot$-fragment of first-order logic if both the predicate interpretation and satisfaction are defined as propositions rather than Boolean functions.
As a second observation, model existence directly implies that valid formulas cannot be unprovable.
Thus, for enumerable theories a single application of $\MP$, rendering enumerable predicates such as deduction stable under double negation,
yields completeness for this formulation of Tarski semantics.
Similarly, $\MPL$ yields the stability of deduction from finite contexts and hence the corresponding form of completeness.
Because $\MP$ is admissible in pCuIC~\cite{PedrotMP}, so are $\MPL$ and the two completeness statements.
For arbitrary theories, completeness becomes equivalent to the law of \emph{Excluded Middle} ($\EM$).

Regarding the second question of our agenda, we show that completeness for the minimal $\to,\forall$-fragment does not depend on additional assumptions by elaborating on a classical proof given in~\cite{SchummCompleteness}
Connectedly, we illustrate how the interpretation of $\bot$ can be relaxed to \emph{exploding models}~\cite{VeldmanExplosion,KrivineCompleteness} admitting a constructive completeness proof for the $\to,\forall,\bot$-fragment.

If, however, Tarski semantics is formulated using a Boolean interpretation for predicates or even a Boolean satisfaction relation, completeness for arbitrary theories becomes equivalent to \emph{both} $\EM$ and $\WKL$.
Since $\EM$ and $\WKL$ are mutually independent in pCuIC, our perspective clarifies that $\WKL$ becomes necessary to treat Boolean models only, and not e.g.\ to treat classical disjunction, as it might be the case for intuitionistic disjunction~\cite{VeldmanExplosion}.

Turning to intuitionistic logic, we discuss analogous relationships for Kripke semantics and a cut-free intuitionistic sequent calculus~\cite{HerbelinCut}.
Again, completeness for the $\to,\forall,\bot$-fragment is equivalent to Markov's Principle while being constructive if restricted to the minimal $\to,\forall$-fragment or employing a relaxed treatment of~$\bot$.
The intuitionistically undefinable connectives $\lor$ and $\exists$ add further complexity~\cite{DankoThesis} and remain untreated in this paper.
As a side note, we explain how the constructivised completeness theorem for intuitionistic logic can be used to implement a semantic cut-elimination procedure.

After considering such model-theoretic semantics, mainly based on embedding the object-logic into the meta-logic, we exemplify two rather different approaches to assigning meaning to formulas, namely algebraic semantics and game semantics.
Differing fundamentally from model-theoretic semantics, both share a constructive rendering of completeness for the full syntax of first-order logic, agnostic to the intuitionistic or classical flavour of the deduction system.

In algebraic semantics, the embedding of formulas into the meta-logic is generalised to an evaluation in algebras providing the structure of the logical connectives.
In this setting, completeness follows from the observation that provability induces such an algebra on formulas.
We discuss intuitionistic and classical logic evaluated in complete Heyting and complete Boolean algebras (cf.~\cite{scott_algebraic_2008}).

Dialogue game semantics as introduced by Lorenzen~\cite{LorenzenDialogues,LorenzenDialogues2}%
, on the other hand, completely disposes of interpreting logical connectives as operations on truth values and instead understand logic as a dialectic game of assertion and argument.
An assertion is considered valid if every sceptic can be convinced through substantive reasoning, i.e. if there is a strategy such that every argument about the assertion can be won.
Hence, game semantics are inherently closer to deduction systems than the previous semantic accounts and in fact a general isomorphism of winning strategies and formal deductions has been established~\cite{SorensenDialogues}.
We adapt this isomorphism such that it can be instantiated to a first-order intuitionistic sequent calculus.

\textbf{Contributions.}
The present paper is an extension of a previous conference publication~\cite{forster2020completeness} in various directions:
Firstly, we extend our previous completeness proof for Tarski semantics restricted to closed formulas in the $\to,\forall,\bot$-fragment to the full syntax with all connectives and allowing free variables in \Cref{sec:completeness_extended}.
Secondly, we deduce compactness from model existence and analyse the connection of Boolean models to $\WKL$ in~\Cref{sec:compactness_WKL}.
Thirdly,  in \Cref{sec:algebras} we give a more detailed treatment of algebraic semantics and discuss a general completeness proof covering all at least intuitionistic natural deduction systems.
Fourthly, in the context of dialogue game semantics (\Cref{sec:dialogues}), we provide a simplified and formal proof of the equivalence of D and E-dialogues, a result hard to reconstruct from the original literature~\cite{FelscherDialogues}.
Finally, we extend our reusable Coq library\footnote{\scriptsize On \url{www.ps.uni-saarland.de/extras/fol-completeness-ext} and hyperlinked with this document} for first-order logic to include all results covered in this paper.

\textbf{Outline.}
In \Cref{sec:prelims}, we begin with some preliminary definitions concerning the syntax of first-order logic, deduction systems, and synthetic computability.
In \Cref{sec:models}, we then analyse completeness for model-theoretic semantics \`a la Tarski (\Cref{sec:tarski}) and Kripke (\Cref{sec:kripke}) and the connections to Weak \Konig's Lemma (\Cref{sec:compactness_WKL}) and Markov's Principle (\Cref{sec:L}).
Subsequently, we give constructive completeness proofs for algebraic semantics (\Cref{sec:algebras}) and dialogue game semantics (\Cref{sec:dialogues}).
We end with a discussion of related and future work in \Cref{sec:discussion} and provide appendices outlining the Coq mechanisation (\Cref{sec:Coq}) and the deduction systems used (\Cref{sec:systems}).

\section{Syntax, Deduction, Computability}
\label{sec:prelims}

We work in a constructive type theory with a predicative hierarchy of type universes above a single impredicative universe $\mathbb{P}$ of propositions.
Assumed type formers are function spaces $X\to Y$, products $X\times Y$, sums $X+Y$, dependent products $\forall x:X.\,F\,x$, and dependent sums $\Sigma\, x:X.\,F\,x$.
The propositional versions of these connectives are denoted by the usual logical symbols ($\to$, $\land$, $\lor$, $\forall$, and $\exists$) in addition to $\top:\Prop$ and $\bot:\Prop$ denoting truth and falsity.\footnote{We use the $\forall$ symbol to denote both dependent product and universal quantification as it is done in Coq itself as well as most of the literature concerned with Coq.}

Basic inductive types are the the unit type $\mathbbm{1} ::= \star$, the Booleans $\Bool::=\btrue\mid\bfalse$, and the natural numbers $\Nat ::=0\mid\natS\,n$ for $n:\Nat$.
Given a type $X$, we further define options $\Opt(X)::=\emptyset\mid \some x$ and lists $\List(X)::=[]\mid x::A$ for $x:X$ and $A:\List(X)$.
On lists we employ the standard notation for membership $x\in A$, inclusion $A\subseteq B$, concatenation $A\app B$, and map $f\,@\,A$.
These notations are shared with vectors $\vec x:X^n$ of fixed length $n:\Nat$.
Possibly infinite collections are expressed by sets $p: X\to \Prop$ with set-theoretic notations like $x\in p$ and $p\subseteq q$.

\subsection{Syntax of First-Order Logic}
\setCoqFilename{FullSyntax}

We represent the terms and formulas of first-order logic as inductive types over a fixed signature $\Sigma=(\Funcs,\Preds)$ specifying function symbols $f:\Funcs$ and predicate symbols $P:\Preds$ together with their arities $|f|:\Nat$ and $|P|:\Nat$.
Variable binding is implemented using de Bruijn indices~\cite{de_bruijn_lambda_1972} well-suited for mechanisation~\cite{AutoSubst2}.

\begin{definition}[][form]
	We define the types $\Term$ of terms and $\Form$ of formulas inductively by
	\begin{small}
		$$t:\Term::=x\mid f\,\vec{t}\hspace{0.3cm}
		\phi,\psi:\Form::=\dot{\bot}\mid P\,\vec{t}\mid\phi\dot{\to}\psi\mid\phi \dot{\land}\psi\mid\phi\dot{\lor}\psi\mid\dot{\forall}\phi\mid\dot{\exists}\phi\hspace{0.3cm}
		x:\Nat, f:\Funcs, P:\Preds$$
	\end{small}
	where the vectors $\vec t$ are of the expected lengths $|f|$ and $|P|$, respectively.
	We set $\dot{\neg}\phi:=\phi\dot{\to}\dot{\bot}$ and isolate the type $\sForm$ of formulas in the $\to,\forall,\bot$-fragment.
\end{definition}

A bound variable is encoded as the number of quantifiers shadowing its relevant binder, e.g. $P\,x\,y\to \forall x.\,\exists y.\, P\,x\,y$ may be represented by $P\,7\,4\dot{\to}\dot{\forall}\,\dot\exists P\,1\,0$.
The variables $7$ and $4$ in this example are called \emph{free} and variables that do not occur freely are called \emph{fresh}.
A formula with no free variables is called \emph{closed}.

\begin{definition}[][subst_form]
	Instantiating with a substitution $\sigma:\Nat \to \Term$ is defined by
	\begin{align*}
		x[\sigma]&~:=~\sigma\,x &
		\dot{\bot}[\sigma]&~:=~\dot{\bot}&
		(\phi\binop\psi)[\sigma] &~:=~\phi[\sigma]\binop\psi[\sigma]\\
		(f\,\vec t\,)[\sigma]&~:=~f\,(\vec t\,[\sigma])&
		(P\,\vec t\,)[\sigma]&~:=~P\,(\vec t\,[\sigma])&
		(\binop\,\phi)[\sigma]&~:=~\binop\,\phi[0;\lambda x.\up(\sigma\,x)]
	\end{align*}
	where $\vec t\,[\sigma]$ denotes $(\lambda t.\,t[\sigma])\,@\,\vec t$, where $t;\sigma$ denotes the substitution mapping $0$ to $t$ and $\natS\,x$ to $\sigma\,x$, where $\up\,t $ denotes $t[\lambda x.\,\natS\,x]$, and where $\binop$ is used as placeholder for the logical connectives and quantifiers, respectively.
\end{definition}

Note that instantiation below a quantifier has to fix the 0 index and shift the substitution by 1 both on input (by using $\_;\_$) and on output (by using $\up\_$).
As two further shorthands, we write $\up\,\phi $ for $\phi[\lambda x.\,\natS\,x]$ and $\phi[t]$ for $\phi[t;\lambda x.\,x]$.
All terminology and notation concerning formulas and substitution carries over to \emph{contexts} $\Gamma:\List(\Form)$ and \emph{theories} $\MCL{T}:\Form\to\Prop$.
For ease of notation we freely identify contexts $\Gamma$ with their induced theory $\lambda\phi.\,\phi\in \Gamma$.

\subsection{Deduction Systems}
\setCoqFilename{FullND}

We represent deduction systems as inductive predicates of type $\List(\Form)\to\Form\to \Prop$ or similar.
The archetypal system is natural deduction (ND), exemplified by an intuitionistic version $\Gamma\vdash \phi$ as defined in \Cref{def:ND} of \Cref{sec:systems}.
Since most rules are standard, we only discuss the quantifier rules in more detail as they rely on the de Bruijn representation of formulas:

$$\infer[\small\textnormal{AI}]{\Gamma\vdash\dot{\forall} \phi}{\up \Gamma\vdash \phi}\hspace{3em}
\infer[\small\textnormal{AE}]{\vphantom{\dot{\forall}}\Gamma\vdash \phi[t]}{\Gamma\vdash \dot \forall \phi}\hspace{3em}
\infer[\small\textnormal{EI}]{\Gamma\vdash \dot{\exists}\phi}{\Gamma\vdash \phi[t]}\hspace{3em}
\infer[\small\textnormal{EE}]{\vphantom{\dot{\forall}}\Gamma\vdash \psi}{\Gamma\vdash\dot{\exists}\phi&\up \Gamma,\phi\vdash \up \psi}
$$

Note that $\up \Gamma,\phi$ is notation for $\phi::\up \Gamma$.
In a shifted context $\up \Gamma$ there is no reference to the variable $0$ which hence plays the role of an arbitrary but fixed individual.
So if $\up \Gamma\vdash \phi$ then we can conclude $\Gamma\vdash\dot \forall \phi$ as expressed by the rule (AI) for $\forall$-introduction.
Similarly, the shifts in the rule (EE) for ${\exists}$-elimination simulate that $\Gamma$ together with $\phi$ instantiated to the witness provided by $\Gamma\vdash\dot{\exists}\phi$ proves $\psi$ and hence admits the conclusion that already $\Gamma\vdash \psi$.
For many proofs it will be helpful to employ fresh variables explicitly as justified by \Cref{lem:named_equiv}, which we state after observing \emph{weakening} and \emph{substitutivity}:

\begin{lemma}[][Weak]
	\label{lem:weak}
	If $\Gamma\vdash \phi$, then $\Delta\vdash \phi$ for all $\Delta\supseteq \Gamma$ and $\Gamma[\sigma]\vdash\phi[\sigma]$ for all $\sigma$.
\end{lemma}

\begin{lemma}[][nameless_equiv_all']
	\label{lem:named_equiv}
	Given $\Gamma$, $\phi$, and $\psi$ one can compute a fresh variable $x$ such that
	\vspace{-0.3cm}
	\begin{multicols}{2}
		\begin{enumerate}
			\item
			$\up \Gamma\vdash \phi~\textit{ iff }~  \Gamma\vdash\phi[x]$ and
			\item
			$\up \Gamma,\phi\vdash \up \psi~\textit{ iff }~\Gamma,\phi[x]\vdash \psi$.
		\end{enumerate}
	\end{multicols}
\end{lemma}

A classical variant $\Gamma\vdash_c \phi$ of the ND system can be obtained without referring to $\dot{\bot}$ by adding the axiom $\Gamma\vdash_c ((\phi\dot\to\psi)\dot\to\phi)\dot\to\phi$ expressing Peirce's law (\Cref{def:CND}).
Then the structural properties stated in the two lemmas above are maintained while the typical classical proof rules become available.

Deduction systems such as intuitionistic ND introduced above naturally extend to theories by writing $\TT\vdash \phi$ if there is a finite
context $\Gamma\subseteq \TT$ with $\Gamma\vdash \phi$.
Then $\TT\vdash \phi$ satisfies proof rules analogous to $\Gamma\vdash \phi$.

\subsection{Synthetic Computability}

Since every function definable in constructive type theory is computable, the standard notions of computability theory can be synthesised by type-level operations~\cite{BAUER20065,ForsterCPP}, eliminating references to a concrete model of computation such as Turing machines, $\mu$-recursive functions, or the untyped lambda calculus.

\begin{definition}
	Let $X$ be a type and $p:X\to \Prop$ be a predicate.
	\begin{itemize}
		\item
		$p$ is \emph{decidable} if there is $f:X \to\Bool$ with
		$\forall x.\,p\,x\leftrightarrow f\,x=\btrue$.
		\item
		$p$ is \emph{enumerable} if there is $f:\Nat\to\Opt(X)$ with $\forall x.\,p\,x\leftrightarrow \exists n.\, f\,n=\some x$.
	\end{itemize}
	These two notions generalise to predicates of higher arity as expected.
	\begin{itemize}
		\item
		$X$ is \emph{enumerable} if there is $f:\Nat\to\Opt(X)$ with $\forall x. \exists n.\, f\,n=\some x$.
		\item
		$X$ is \emph{discrete} if equality $\lambda xy.x=y$ on $X$ is decidable.
		\item
		$X$ is a \emph{data type} if it is both enumerable and discrete.
	\end{itemize}
	
\end{definition}

We assume that the components $\Funcs$ and $\Preds$ of our fixed signature $\Sigma$ are data types.
Then applying the terminology to the syntax and deductions systems introduced in the previous sections leads to the following observations.

\setCoqFilename{FOL}
\begin{fact}[][enumT_form]
	\label{fact:data}
	\Term and \Form are data types and $\Gamma\vdash\phi$ and $\Gamma\vdash_c\phi$ are enumerable.
\end{fact}

\begin{proofqed}
	By the techniques discussed in \cite{ForsterCPP}, e.g. Fact 3.19.
\end{proofqed}

The standard model-theoretic completeness proofs analysed in \Cref{sec:models} require the assumption of Markov's Principle.
A proposition $P:\Prop$ is called \emph{stable} if $\neg\neg P\to P$ and, analogously, a predicate $p:X\to\Prop$ is called stable if $p\,x$ is stable for all $x$.
A synthetic version of Markov's Principle states that satisfiability of Boolean sequences is stable~(cf.~\cite{MannaaMP}):
$$\MP:=\forall f:\Nat\to\Bool.\,\neg\neg(\exists n.\,f\,n=\btrue)\to \exists n.\,f\,n=\btrue$$
Note that \MP is trivially implied by Excluded Middle $\EM:=\forall P:\Prop.\,P\lor\neg P$.
Moreover, \MP regulates the behaviour of computationally tractable predicates:

\setCoqFilename{Markov}
\begin{fact}[][MP_enum_stable_iff]
	\label{fact:MP}
	\MP holds iff all enumerable predicates on discrete types are stable.
\end{fact}

\begin{proofqed}
  The direction from left to right is Fact 2.18 in  \cite{ForsterCPP}.
  For the reverse direction assume that enumerable predicates on discrete types are stable.
  Let $f : \Nat \to \Bool$ and let $p :\mathbbm{1}\to \Prop$ be defined by $p \,x := \exists n.\, f\, n = \btrue$.
  The predicate $p$ is enumerable by $f \,n := \textbf{if } f\, n \textbf{ then } \some \star \textbf{ else } \emptyset$.
  Stability of $p$ is now equivalent to $\neg\neg(\exists n.\, f\, n = \btrue) \to (\exists n.\, f\, n = \btrue)$.
\end{proofqed}
\setCoqFilename{FOL}
As a consequence of \Cref{fact:data} and \Cref{fact:MP}, \MP implies that the deduction systems $\Gamma\vdash\phi$ and $\Gamma\vdash_c\phi$ are stable.
In fact, only these stabilities are required for the standard model-theoretic completeness proofs discussed in the next section and they are equivalent to $\MPL$, a version of Markov's Principle stated for the call-by-value $\lambda$-calculus $\L$~\cite{Plotkin75,ForsterL} and its halting problem $\eva$:
$$\MPL := \forall s.~\neg\neg\eva{s} \to \eva{s}$$
We will prove the following in \Cref{sec:L}:

\begin{lemma}\label{lem:MPL-equivs}
  $\MPL$, stability of $\Gamma\vdash\phi $and stability of $\Gamma\vdash_c\phi $ are all equivalent.
\end{lemma}

\section{Model-Theoretic Semantics}
\label{sec:models}

The first variant of semantics we consider is based on the idea of interpreting terms as objects in a model and embedding the logical connectives into the meta-logic.
A formula is considered valid if it is satisfied by all models.
The simplest case is Tarski semantics, coinciding with classical deduction via Henkin's completeness proof factoring through a (constructive) model-existence theorem~\cite{HenkinCompleteness}.
Kripke semantics, coinciding with intuitionistic deduction, add more structure by connecting several models through an accessibility relation and admit a simpler completeness proof using a universal model.
In this section, we only consider formulas $\phi:\sForm$ in the $\to,\forall,\bot$-fragment if not stated otherwise.

\subsection{Tarski Semantics}
\setCoqFilename{GenTarski}
\label{sec:tarski}

\begin{definition}[][interp]
	A \emph{(Tarski) model} $\MM$ over a domain $D$ is a pair of functions
	$$\_^\MM~:~\forall f:\Funcs.\,D^{|f|}\to D\hspace{5em}
	\_^\MM~:~\forall P:\Preds.\,D^{|P|}\to \Prop.$$
	\emph{Assignments} $\rho:\Nat\to D$ are extended to \emph{term evaluations} $\hat\rho:\Term\to D$ by $\hat\rho \,x:=\rho\,x$ and $\hat\rho\,(f\,\vec t\,):=f^\MM\,(\hat{\rho}\,@\,\vec t\,)$
	and to formulas via the relation $\MM\vDash_\rho \phi$ defined by
	\begin{align*}
		\MM\vDash_\rho \dot{\bot}&~:=~\bot&
		\MM\vDash_\rho \phi\dot\to\psi &~:=~\MM\vDash_\rho\phi\to \MM\vDash_\rho\psi\\
		\MM\vDash_\rho P\,\vec t\,&~:=~P^\MM\,(\hat{\rho}\,@\,\vec t\,)&
		\MM\vDash_\rho\dot{\forall}\,\phi&~:=~\forall a:D.\,\MM\vDash_{a;\rho} \phi
	\end{align*}
	where the assignment $a;\rho$ maps $0$ to $a$ and $\natS\,x$ to $\rho\,x$.
	We write $\MM\vDash\phi$ if $\MM\vDash_\rho \phi$ for all $\rho$.
	$\MM$ is called \emph{classical} if it validates all instances of Peirce's law, i.e. $\MM\vDash ((\phi\dot\to\psi)\dot\to\phi)\dot\to\phi$ for all $\phi,\psi:\sForm$.
	We write $\MM\vDash_\rho \TT$ if $\MM_\rho\vDash \phi$ for all $\phi\in\TT$ and $\TT\vDash \phi$ if $\MM\vDash_\rho \phi$ for every classical $\MM$ and $\rho$ with $\MM\vDash_\rho \TT$.
\end{definition}

We first show that the classical deduction system $\Gamma\vdash_c\phi$ (restricted to the considered $\to,\forall,\bot$-fragment) is \emph{sound} for Tarski semantics.

\begin{fact}[][Soundness']\label{tarski_soundness}
	$\Gamma\vdash_c\phi$ implies $\Gamma\vDash \phi$.
\end{fact}

\begin{proofqed}
	By induction on $\Gamma\vdash_c\phi$ similar to the soundness proof in~\cite[Fact 3.14]{ForsterCPP}.
	The classical Peirce axioms $\Gamma\vdash_c ((\phi\dot\to\psi)\dot\to\phi)\dot\to\phi$ are sound given that we only consider classical models.
\end{proofqed}

Formally, \emph{completeness} denotes the converse property, i.e. that $\Gamma\vDash\phi$ implies $\Gamma\vdash_c \phi$.
We now outline a Henkin-style completeness proof for $\Gamma\vdash_c\phi$ based on the presentation by Herbelin and Ilik~\cite{HerbelinHenkin}.
The main idea is to factor through a model existence theorem, stating that every consistent context is satisfied by a syntactic model.
The model existence theorem in turn is based on a theory extension lemma generalising the role of $\dot{\bot}$ to an arbitrary substitute $\phi_\bot$:

\setCoqFilename{GenConstructions}
\begin{lemma}[][construct_construction]
	\label{lem:extension}
	For every closed formula $\phi_\bot$ and closed $\TT$ there is $\TT'\supseteq \TT$ with:
	\begin{enumerate}
		\item
		$\TT'$ maintains $\phi_\bot$-consistency, i.e. $\TT\vdash_c \phi_\bot$ whenever $\TT'\vdash_c \phi_\bot$.
		\item
		$\TT'$ is deductively closed, i.e. $\phi\in \TT'$ whenever $\TT' \vdash_c \phi$.
		\item
		$\TT'$ respects implication, i.e. $\phi\dot{\to}\psi\in \TT'$ iff $\phi \in \TT'\to \psi \in \TT'$.
		\item
		$\TT'$ respects universal quantification, i.e. $\dot{\forall}\phi\in \TT'$ iff $\forall t.\,\phi[t] \in \TT'$.
	\end{enumerate}
\end{lemma}

\begin{proof}
	We fix an enumeration $\phi_n$ of $\sForm$ such that $x$ is fresh for $\phi_n$ if $x\ge n$.
	The extension can be separated into three steps, all maintaining $\phi_\bot$-consistency:
	
	\begin{enumerate}
		\item[a.]
		$\MCL{E}\supseteq \TT$ which is \emph{exploding}, i.e. $(\phi_\bot\dot{\to}\phi) \in \MCL{E}$ for all closed $\phi$.
		\item[b.]
		$\MCL{H}\supseteq \MCL{E}$ which is \emph{Henkin}, i.e. $(\phi_n[n]\dot{\to}\dot{\forall}\phi_n)\in \MCL{H}$ for all $n$.
		\item[c.]
		$\Omega \supseteq \MCL{H}$ which is \emph{maximal}, i.e. $\phi\in\Omega$ whenever $\Omega,\phi\vdash_c\phi_\bot$ implies $\Omega\vdash_c\phi_\bot$.
	\end{enumerate}
	
	Note that being exploding allows to use $\phi_\bot$ analogously to $\dot{\bot}$ and that being Henkin ensures that there is no mismatch between the provability of a universal formula and all its instances.
	We first argue why $\Omega$ satisfies the claims (1)-(4) of the extension lemma.
	
	\begin{enumerate}
		\item
		$\Omega$ is a $\phi_\bot$-consistent extension of $\TT$ since all steps maintain $\phi_\bot$-consistency.
		\item
		Let $\Omega\vdash_c\phi$ and assume $\Omega,\phi\vdash_c\phi_\bot$, so $\Omega\vdash_c\phi_\bot$.
		Thus $\phi\in \Omega$ per maximality.
		\item
		The first direction is immediate as $\Omega$ is deductively closed.
		We prove the converse using maximality, so assume $\Omega,\phi\dot{\to}\psi\vdash_c\phi_\bot$.
		It suffices to show that $\Omega \vdash_c\phi$ since then $\phi \in \Omega$, $\psi\in \Omega$, and ultimately $\Omega\vdash_c\phi_\bot$ follow.
		$\Omega \vdash_c\phi$ can be derived by proof rules for $\phi_\bot$ analogous to the ones for $\dot{\bot}$.
		\item
		The first direction is again immediate by $\Omega$ being deductively closed and the converse exploits that $\Omega$ is Henkin as follows.
		Suppose $\forall t.\,\phi[t] \in \Omega$ and let $\phi$ be $\phi_n$ in the given enumeration.
		Then in particular $\phi_n[n]\in \Omega$ and since $\Omega$ is Henkin also $\phi_n[n]\dot{\to}\dot{\forall}\phi_n\in \Omega$
		which is enough to derive $\dot{\forall}\phi\in \Omega$.
	\end{enumerate}
	
	We now discuss the three extension steps separately:
	
	\begin{enumerate}
		\item[a.]
		Since the requirement is unconditional, we just add all needed formulas:
		$$\MCL{E}:=\TT\cup \{\phi_\bot\dot{\to}\phi\mid \phi \textnormal{ closed} \}$$
		We only have to argue that $\MCL{E}$ maintains $\phi_\bot$-consistency over $\TT$.
		So suppose $\MCL{E}\vdash_c \phi_\bot$, meaning that $\Gamma\vdash_c\phi_\bot$ for some $\Gamma \subseteq \MCL{E}$.
		We show that all added instances of explosion for $\phi_\bot$ in $\Gamma$ can be eliminated.
		Indeed, for $\Gamma=\Delta,\phi_\bot\dot{\to}\phi$ we have $\Delta\vdash_c(\phi_\bot\dot{\to}\phi)\dot{\to}\phi_\bot$ and hence $\Delta\vdash_c \phi_\bot$ by the Peirce rule.
		Thus by iteration there is $\Gamma'\subseteq \TT$ with $\Gamma'\vdash_c\phi_\bot$, justifying $\TT\vdash_c\phi_\bot$.
		\item[b.]
		As above, to make \MCL{E} Henkin we just add all necessary Henkin-axioms
		$$\MCL{H}:=\MCL{E}\cup \{\phi_n[n]\dot{\to}\dot{\forall}\phi_n\mid n:\Nat \}$$
		and justify that the extension maintains $\phi_\bot$-consistency.
		So let $\Gamma\vdash_c\phi_\bot$ for some $\Gamma \subseteq \MCL{H}$, we again show that all added instances can be eliminated.
		Hence suppose $\Gamma=\Delta,\phi_n[n]\dot{\to}\dot{\forall}\phi_n$.
		One can show that in a context $\Delta'$ extending $\Delta$ by suitable instances of $\phi_\bot$-explosion one can derive $\Delta'\vdash_c \phi_\bot$.
		In this derivation one exploits that $n$ is fresh for $\phi_n$ and that the input theory $\MCL{E}$ is closed.
		Thus ultimately $\MCL{E}\vdash_c \phi_\bot$.
		\item[c.]
		The last step maximises \MCL{H} by adding all formulas maintaining $\phi_\bot$-consistency:
		$$\Omega_0 :=\MCL H\hspace{1em}
		\Omega_{n+1}:=\Omega_n\cup\{\phi_n\mid \Omega_n,\phi_n\vdash_c \phi_\bot\textnormal{ implies }\Omega_n\vdash_c \phi_\bot \}\hspace{1em}
		\Omega :=\bigcup_{n:\Nat}\Omega_n$$
		Note that $\Omega$ maintains $\phi_\bot$-consistency over all $\Omega_n$ and hence \MCL{H} by construction so it remains to justify that $\Omega$ is maximal.
		So suppose $\Omega,\phi_n\vdash_c\phi_\bot$ implies $\Omega\vdash_c\phi_\bot$, we have to show that $\phi_n\in\Omega$.
		This is the case if the condition in the definition of $\Omega_{n+1}$ is satisfied, so let $\Omega_n,\phi_n\vdash_c \phi_\bot$.
		Then by the assumed implication $\Omega\vdash_c\phi_\bot$ and since $\Omega$ maintains $\phi_\bot$-consistency over $\Omega_n$ also $\Omega_n\vdash_c\phi_\bot$ as required.
		\qed
	\end{enumerate}
\end{proof}

\setCoqFilename{GenCompleteness}

The generalisation via the falsity substitute $\phi_\bot$ will become important later, for now the instance $\phi_\bot:=\dot{\bot}$ suffices.
Also note that in usual jargon the extension $\TT'$ of a consistent theory $\TT$ is called \emph{maximal consistent}, as no further formulas can be added to $\TT'$ without breaking consistency.

Maximal consistent theories $\TT$ give rise to equivalent \emph{syntactic models} $\MM_\TT$ over the domain \Term of terms by setting $f^{\TT}\,\vec t:=f\,\vec t$ and $P^{\TT}\,\vec t:=(P\,\vec t\in \TT)$.
We then observe that $\MM_\TT\vDash_\sigma \phi$ iff $\phi[\sigma]\in\TT$ for all substitutions $\sigma$ by a straighforward induction on $\phi$ using the properties stated in \Cref{lem:extension}.
Hence in particular $\MM_\TT\vDash_\id \phi$ iff $\phi\in\TT$ for the identity substitution $\id\,x:=x$.
From this observation we directly conclude the model existence theorem:

\setCoqFilename{GenCompleteness}
\begin{theorem}[][model_bot_correct]
	\label{thm:model}
	Every closed and consistent theory is satisfied in a classical model.
\end{theorem}

\begin{proofqed}
	Let \TT be closed and consistent and let $\TT'$ be its extension per~\Cref{lem:extension} for $\phi_\bot:=\dot{\bot}$.
	To show $\MM_{\TT'}\vDash_\id \TT$, let $\phi\in \TT$, hence $\phi\in \TT'$.
	Then since $\MM_{\TT'}$ is equivalent to $\TT'$ we conclude $\MM_{\TT'}\vDash_\id \phi$ as desired.
	Finally, $\MM_{\TT'}$ is classical due to (2) of \Cref{lem:extension}.
\end{proofqed}

The model existence theorem yields completeness up to double negation:

\begin{fact}[][semi_completeness_standard]
	\label{thm:quasi}
	$\TT\vDash \phi$ implies $\neg\neg (\TT\vdash _c \phi)$ for arbitrary $\TT$ and $\phi$.
\end{fact}

\begin{proofqed}
	First, suppose that $\TT\vDash \phi$ for closed $\TT$ and $\phi$ and assume
  $\TT\not\vdash _c \phi$ which is equivalent to $\TT,\dot\neg \phi$ being
  consistent. But then there must be a model of $\TT,\dot\neg \phi$ in conflict
  to the assumption $\TT\vDash \phi$.

  To extend this result to arbitrary $\TT$ and $\varphi$ one can simply close
  them by replacing all free variables with fresh constants. We spell out the
  details of this construction in \Cref{lem:closing}.
\end{proofqed}

In fact, the remaining double negation elimination turns out to be necessary:

\begin{fact}[][completeness_standard_stability]
	\label{fact:equivc}
	Completeness of $\Gamma\vdash_c\phi$ is equivalent to stability of $\Gamma\vdash_c\phi$.
\end{fact}

\begin{proofqed}
	Assuming stability, \Cref{thm:quasi} directly yields the completeness of $\Gamma\vdash_c\phi$.
	Conversely, assume completeness and let $\neg\neg (\Gamma\vdash _c \phi)$.
	Employing completeness, to get $\Gamma\vdash _c \phi$ it suffices to show $\Gamma,\dot\neg \phi\vDash \dot{\bot}$, so suppose $\MM\vDash_\rho \Gamma,\dot\neg\phi$ for some $\MM$ and $\rho$.
	As we now aim at a contradiction, we can turn $\neg\neg (\Gamma\vdash _c \phi)$ into $\Gamma\vdash _c \phi$ and therefore obtain $\Gamma\vDash _c \phi$ by soundness, a conflict to $\MM\vDash_\rho \Gamma,\dot\neg\phi$.
\end{proofqed}

Hence, we can characterise completeness of classical ND as follows.

\setCoqFilename{Analysis}
\begin{theorem}
	\label{thm:characterisations}
	\begin{enumerate}
		\coqitem[completeness_context_iff_MPL]
    Completeness of $\Gamma\vdash_c\phi$ is equivalent to $\MPL$.
		\coqitem[completeness_enum_iff_MP]
		Completeness of $\TT\vdash_c\phi$ for enumerable \TT is equivalent to \MP.
		\coqitem[completeness_iff_XM]
		Completeness of $\TT\vdash_c\phi$ for arbitrary \TT is equivalent to \EM.
	\end{enumerate}
\end{theorem}
\setCoqFilename{GenCompleteness}

\begin{proof}
	\begin{enumerate}
		\item
		By Fact~\ref{fact:equivc} completeness is equivalent to the stability of $\Gamma\vdash_c\phi$ which is shown equivalent to $\MPL$ in~\Cref{sec:L}.
		\item
		$\TT\vdash_c\phi$ for enumerable \TT is enumerable, hence stable under \MP and thus complete per \Cref{thm:quasi}.
		For the converse, assume a function $f:\Nat\to\Bool$ and consider $\TT:=(\lambda \phi.\,\phi=\dot{\bot}\land \exists n.\,f\,n=\btrue)$.
		Since $\TT$ is enumerable, completeness yields that $\TT\vDash\dot{\bot}$ is equivalent to $\TT\vdash_c\dot{\bot}$ which in turn is equivalent to $\exists n.\,f\,n=\btrue$.
		Then since $\TT\vDash\dot{\bot}$ is stable so must be $\exists n.\,f\,n=\btrue$.
		\item
		\EM particularly implies that $\TT\vdash_c\phi$ is stable and hence complete.
		Conversely given a proposition $P:\Prop$, completeness for $\TT:=(\lambda \phi.\,\phi=\dot{\bot}\land P)$ yields the stability of $P$ with an argument as in (2).
		\qed
	\end{enumerate}
\end{proof}

Having analysed the usual Henkin-style completeness proof, we now turn to its constructivisation.
The central observation is that completeness already holds constructively for the minimal $\to,\forall$-fragment, by an elaboration of the classical proof for the minimal fragment given in~\cite{SchummCompleteness}.
To this end, we further restrict the deduction system and semantics to the minimal fragment and prove completeness via a suitable form of model existence.

\begin{lemma}[][model_fragment_correct]
	In the $\to,\forall$-fragment, for closed \TT and $\phi$ there is a classical model~$\MM$ and an assignment $\rho$ such that (1)~$\MM\vDash_\rho \TT$ and (2)~$\MM\vDash_\rho\phi$ implies $\TT\vdash_c \phi$.
\end{lemma}

\begin{proofqed}
	Let $\TT'$ be the extension of $\TT$ for $\phi_\bot:=\phi$.
	As before, we have $\MM_{\TT'}\vDash_\id \TT'$.
	So now let $\MM_{\TT'}\vDash_\id \phi$, then $\phi\in \TT'$ and $\TT\vdash_c\phi$ by (1) of \Cref{lem:extension}.
\end{proofqed}

\begin{corollary}[][semi_completeness_fragment]
	In the $\to,\forall$-fragment, $\Gamma\vDash\phi$ implies $\Gamma\vdash_c\phi$ for closed $\Gamma$ and $\phi$.
\end{corollary}

As opposed to completeness for fomulas incorporating $\dot{\bot}$, completeness in the minimal fragment does not rely on consistency requirements.
Consequently, if these requirements are eliminated by allowing models treating inconsistency more liberal, completeness for formulas with $\dot{\bot}$ can be established constructively~(cf.~\cite{VeldmanExplosion,KrivineCompleteness}).

So we now turn back to the $\to,\forall,\bot$-fragment and define a satisfaction relation $\MM\vDash_\rho^A\phi$ for arbitrary propositions $A$ with the relaxed rule $(\MM\vDash_\rho^{A} \dot\bot):=A$.
A model $\MM$ is \emph{$A$-exploding} if $\MM\vDash^A\dot{\bot}\to \phi$ for all $\phi$ and \emph{exploding} if it is $A$-exploding for some choice of $A$.
Note that $A:=\top$ and $P^\MM\,\vec t:=\top$ in particular yields an exploding model satisfying all formulas, hence accommodating inconsistent theories. 
This leads to the following formulation of model existence.

\begin{lemma}[][model_bot_correct]
	\label{lem:exploding}
	For every closed theory \TT there is an exploding classical model $\MM$ and an assignment $\rho$ such that (1)~$\MM\vDash^A_\rho \TT$ and (2)~$\MM\vDash_\rho^A\dot\bot$ implies $\TT\vdash_c \dot\bot$.
\end{lemma}

\begin{proofqed}
	Let $\TT$ be closed and let $\TT'$ be its extension for $\phi_\bot:=\dot{\bot}$.
	We set $A:=\dot{\bot}\in \TT'$ and observe that the syntactic model $\MM_{\TT'}$ still coincides with $\TT'$, i.e. $\MM_{\TT'}\vDash_\sigma^A \phi$ iff $\phi[\sigma]\in\TT'$.
	Hence we have (1) $\MM_{\TT'}\vDash_\id ^A \TT$.
	Moreover, $\MM_{\TT'}$ is $A$-exploding since proving $\MM_{\TT'}\vDash^A_\sigma\dot{\bot}\to \phi$ in this case means to prove that $\dot{\bot}\dot\to\phi[\sigma]\in \TT'$, a straightforward consequence of $\TT'$ being deductively closed.
	Finally, (2) follows from (1) of \Cref{lem:extension} as seen before.
\end{proofqed}

We write $\Gamma\vDash_e\phi$ if $\MM\vDash^A_\rho\phi$ for all $A:\Prop$ and $A$-exploding $\MM$ and $\rho$ with $\MM\vDash^A_\rho\Gamma$ and finally establish completeness with respect to exploding models:

\begin{fact}[][completeness_expl]
	$\Gamma\vDash_e\phi$ implies $\Gamma\vdash_c\phi$ for closed $\Gamma$ and $\phi$.
\end{fact}

\begin{proofqed}
	Let $\Gamma\vDash_e\phi$, then $\Gamma,\dot{\neg}\phi\vdash_c\dot\bot$ follows by \Cref{lem:exploding} for $\TT := \Gamma,\dot{\neg}\phi$.
\end{proofqed}

\subsection{Completeness Extended to Full Syntax and Free Variables}
\label{sec:completeness_extended}

\setCoqFilename{DeMorgan}

The completeness statements discussed in the previous section impose syntactic limitations in two ways: we only considered formulas belonging to the $\to,\forall,\bot$-fragment and did not explain the treatment of free variables underlying \Cref{coq:semi_completeness_standard}.
Both of these shortcomings are addressed in this section.

First, we show how completeness for the full syntax $\Form$ can be reduced to completeness for the fragment $\Form^*$.
To this end, we formally distinguish the deduction systems $\Gamma \vdash_c^*\phi$ and $\Gamma \vdash_c\phi$ and satisfaction relations $\MM \vDash_\rho^* \phi$ and $\MM \vDash_\rho \phi$ involving formulas from $\Form^*$ and $\Form$, respectively.
As mentioned earlier, the classical deduction system $\Gamma \vdash_c^*\phi$ is already suitable to encode the missing connectives via the usual classical equivalents.
However, if we extend the Tarski semantics $\MM \vDash_\rho^* \phi$ to formulas $\phi : \Form $ in the natural way, in particular by setting
$$\MM\vDash_\rho \phi\dot\lor\psi ~:=~\MM\vDash_\rho\phi\lor \MM\vDash_\rho\psi\hspace{1cm}
\MM\vDash_\rho\dot{\exists}\,\phi~:=~\exists a:D.\,\MM\vDash_{a;\rho} \phi$$

then classical logic on the meta-level becomes necessary to tame the constructively stronger notions of disjunction and existence.

For ease of readability, we identify formulas in $\Form^*$ with their identity embedding into $\Form$.
The converse encoding of $\Form$ into $\Form^*$ is defined as follows:

\begin{definition}[][DM]
	We define the de Morgan translation $\phi^M$ from $\Form$ to $\Form^*$ by
	$$(\phi\dot\land\psi)^M := \dot{\neg}(\phi^M\dot{\to}\dot{\neg}\psi^M)\hspace{0.7cm}
		 (\phi\dot\lor\psi)^M := \dot{\neg}\phi^M\dot{\to}\psi^M\hspace{0.7cm}
		 (\dot\exists\phi)^M := \dot{\neg}\dot{\forall}\dot{\neg}\phi^M$$	
	in the crucial cases and with the remaining syntax just recursively traversed.
\end{definition}

We verify that the deduction system indeed cannot distinguish formulas from their de Morgan translations:

\begin{lemma}[][DM_prv]
	\label{DM_prv}
	$\Gamma \vdash_c \phi$ iff $\Gamma\vdash_c \phi^M$ and in particular $\Gamma \vdash_c \phi$ iff $\Gamma^M\vdash_c^* \phi^M$.
\end{lemma}

\begin{proofqed}
	The first equivalence is by induction on $\phi$ with $\Gamma$ generalised with the backwards directions relying on the classical (P) rule as expected.
	The implication from $\Gamma \vdash_c \phi$ to $\Gamma^M\vdash_c^* \phi^M$ is by induction on $\Gamma \vdash_c \phi$ employing that substitution commutes with the de Morgan translation.
	The converse implication follows with the first equivalence since all fragment deductions can be replayed in the full system.
\end{proofqed}

Turning to the semantics, the deductive equivalence can be mimicked when assuming classical logic.

\begin{lemma}[][DMT_sat]
	\label{lem:DMT_sat}
	Given \EM, we have $\MM \vDash_\rho \phi$ iff $\MM \vDash_\rho^* \phi^M$ for all $\MM$ and $\rho$.
\end{lemma}

\begin{proofqed}
	By induction on $\phi$ with $\rho$ generalised, using $\EM$ to get from $\phi^M$ to $\phi$.
\end{proofqed}

\begin{corollary}[][DMT_valid]
	\label{corollay:DMT_valid}
	Given \EM, $\TT \vDash \phi$ implies $\TT^M \vDash^* \phi^M$ for all $\TT$ and $\phi$.
\end{corollary}

Therefore, we can conclude a completeness statement as follows.

\begin{theorem}[][full_completeness]
	\label{thm:full_completeness}
	Given \EM, $\TT\vDash\phi$ implies $\TT\vdash_c\phi$ for closed $\TT$ and $\phi$ in $\Form$.
\end{theorem}

\begin{proofqed}
	By composing \Cref{corollay:DMT_valid}, \Cref{thm:characterisations}, and \Cref{DM_prv}.
\end{proofqed}

Note that this concluding theorem requires full classical logic as analysed before in \Cref{thm:characterisations}.
Moreover, so does the general statement of \Cref{lem:DMT_sat}:

\begin{fact}[][DMT_sat_back]
	If $\MM \vDash_\rho \phi$ iff $\MM \vDash_\rho^* \phi^M$ for all $\Sigma$, $\MM$ and $\rho$, then $\EM$ holds.
\end{fact}

\begin{proof}
	Given a proposition $P$, we instantiate the assumed equivalence with the signature containing only a single propositional variable $p$, the model $\MM$ on domain $\mathbbm{1}$ interpreting $p$ as $P$, and the constant environment $\rho\,n:=\star$.
	Then the claim $P\lor \neg P$ can be expressed as $\MM \vDash_\rho p\dot{\lor}\dot{\neg}p$.
	By the assumed equivalence, we just need to prove $\MM \vDash_\rho^* (p\dot{\lor}\dot{\neg}p)^M$ which reduces to the tautology $\neg P \to \neg P$.
\end{proof}

However, we suspect that \Cref{corollay:DMT_valid} actually requires only a weaker assumption due to the restriction to classical models in the relation $\TT\vDash\phi$.

Secondly, to extend the completeness results to open theories, we show that the free variables of such
theories may be replaced with fresh constants, thereby closing them, without
changing any of their consequences up to substitution of constants. Note that as
our construction requires $\varphi_\bot$ to be a closed formula, which means
shifting the free variables to guarantee the presence of
countably many unused free variables, such as in~\cite{HerbelinHenkin}, will not
be sufficient.

For this, we take $\Sigma_c$ to be the extension of a signature $\Sigma$ with
countably many new constants $c_{-} : \Nat \to \MCL{F}_{\Sigma_c}$ and $\Uparrow
\,: \Form_\Sigma \to \Form_{\Sigma_c}$ to be the associated embedding of formulas.
We then define a dropping operation $\Downarrow^m : \Form_{\Sigma_c} \to
\Form_{\Sigma}$ which replaces occurrences of $c_n$ with the variable $m + n$.
This index is necessary as variables have to be shifted when moving below
quantifiers to refer to the intended free variable. We can now extend the
constructivised completeness result to open theories and formulas.

\setCoqFilename{GenCompleteness}
\begin{lemma}[][strong_completeness_standard]\label{lem:closing}
  If $\TT \vDash \varphi$ entails $\neg \neg (\TT \vdash_c \varphi)$ for closed
  $\TT$ and $\varphi$, then this can be extended to arbitrary $\TT$ and $\varphi$.
\end{lemma}
\begin{proof}
  For this we need the following intermediary facts which are proven per
  induction on the formula and the derivation, respectively.
  \begin{enumerate}[(1)]
  \item For any $\varphi$, $\subst{\varphi}{c_{-}}$ is closed and
    $\sigdropn{0}{( \subst{( \siglift{\varphi} )}{c_{-}} )} = \varphi$
  \item For any $\TT$ and $\varphi$, $\TT \vDash \varphi \to
    \siglift{\TT} \vDash\, \siglift{\varphi}$ and $\varphi$, $\TT \vdash_c \varphi \to
    \sigdropn{m}{\TT} \vdash_c \sigdropn{m}{\varphi}$
  \end{enumerate}
  We may then derive the claim by the following sequence of implications:
  \begin{align*}
    & \TT \vDash \varphi & \\
    \to \quad & \siglift{\TT}\, \vDash \siglift{\varphi} & (2)  \\
    \to \quad & \subst{(\siglift{\TT})}{c_{-}} \vDash \subst{(\siglift{\varphi})}{c_{-}} & \text{Weakening under substitution} \\
    \to \quad & \neg \neg (\subst{(\siglift{\TT})}{c_{-}} \vdash_c \subst{(\siglift{\varphi})}{c_{-}}) & \text{Assumption}, (1) \\
    \to \quad & \neg \neg (\sigdrop{( \subst{(\siglift{\TT})}{c_{-}} )} \vdash_c \sigdrop{( \subst{(\siglift{\varphi})}{c_{-}} )}) & (2) \\
    \to \quad & \neg \neg \TT \vdash_c \varphi & (1)&&{~~~\qed} 
  \end{align*}
\end{proof}

\subsection{Compactness and Weak \Konig's Lemma}
\label{sec:compactness_WKL}
\setCoqFilename{WKL}

We have proved the model existence theorem for classical models fully constructively and deduced completeness of provability in arbitrary theories using $\EM$.
Recall that we defined both the interpretation of atoms in a model and the satisfiability relation to be propositional, as is most natural in our setting.
In classical presentations, defining satisfiability as relation is equivalent to a definition as Boolean function by relying on \EM.
In type theory however, identifying relations with Boolean functions needs choice axioms.
In this section, we analyse the connection between Weak \Konig's Lemma and the model existence theorem w.r.t. models employing Boolean interpretations of symbols and Boolean satisfaction relations.

\begin{definition}[][omniscient]
  We call a classical model \emph{decidable} if its predicate interpretation is decidable, and \emph{omniscient} if the satisfiability relation is.
\end{definition}

\begin{lemma}[][omniscient_to_classical]
  Non-empty omniscient models are decidable.
  Decidable models with finite domain are omniscient.
\end{lemma}

Neither classicality nor decidability imply omniscience, and nor do they imply each other:
The standard model for (Heyting) arithmetic in type theory is decidable (because equality on natural numbers is), but can neither be proved classical (because type theory is constructive), nor omniscient (because of Hilbert's tenth problem).
A model where the domain ranges over Turing machines for a signature with exactly one unary predicate denoting Turing machine halting is not decidable (and thus not omniscient), but classical under the assumption of $\EM$.

In classical reverse mathematics, where one assumes classical logic but only restricted forms of set existence axioms and induction, it is a well-known theorem that the model existence theorem and the compactness theorem are equivalent to Weak \Konig's Lemma (\WKL)~if satisfiability is a Boolean function~\cite{simpson2009subsystems}, i.e. if model-existence is phrased for omniscient models.
A similar analysis is carried out in publications on constructive reverse mathematics, where one does not assume classical logic, but the axiom of countable or even dependent choice.
However, we are only aware of an analysis for the compactness theorem for propositional logic rather than first-order logic, which is equivalent to \WKL~for decidable trees~\cite{dienerConstructiveReverseMathematics2020}.

We start by deducing the compactness theorem for various classes of models.
The compactness theorem already appeared in Gödel's seminal paper~\cite{GodelCompleteness} and states that a theory $\TT$ has a model if every context $\Gamma \subseteq \TT$ has a model.
It is in fact easy to deduce the compactness theorem for classes of models which are at least classical from model existence:

\begin{theorem}[][modex_compact]
  Let $C$ be a predicate on models s.t. $C \MM$ implies that $\MM$ is classical.
  Then the model existence theorem for models in $C$ implies the compactness theorem for models in $C$.
  Formally, for every signature $\Sigma$:
  \begin{align*}
    &(\forall \TT.\, \TT \not \vdash_c \dot\bot \to \exists \MM.\, C \MM \land \MM \vDash \TT  ) \\ \to &(\forall \TT.\,(\forall \Gamma \subseteq \TT.\, \exists \MM.\, C \MM \land \MM \vDash \Gamma) \to \exists \MM.\, C \MM \land \MM \vDash \TT)
  \end{align*}
\end{theorem}
\begin{proofqed}
  Let $\TT$ be a theory.
  It suffices to prove that $\MM \vDash \Gamma$ for $\Gamma\subseteq \TT$ and $C \MM$ implies $\Gamma \not\vdash_c \dot\bot$, which follows directly using \Cref{tarski_soundness}.
\end{proofqed}

\begin{corollary}[][compact_standard]
  If every $\Gamma \subseteq \TT$ has a classical model, $\TT$ has a classical model.
\end{corollary}

We continue by proving that compactness for decidable models implies $\WKL$.
We introduce $\WKL$ formally:

\newcommand\tree\tau
\begin{definition}[][WKL]
  A \emph{binary tree} is a non-empty and prefix-closed predicate\\ $\tree : \List(\Bool) \to \Prop$, i.e.
  $\tau \,[]$ holds and $\tree\, v$ implies $\tree\, u$ for all prefixes $u$ of $v$.

  A binary tree $\tau$ is \emph{infinite} if $\forall k.\, \exists u.\, \tau u \land | u | \geq k$ and $\tau$ has an \emph{infinite path} if $\exists f : \Nat \to \Bool.\, \forall n.\, \tau\, [f\, 0, \dots, f\, n]$.

  $\WKL$ states that every infinite binary tree has an infinite path.
  $\WKL_{\mathcal{D}}$ states that every \emph{decidable} infinite binary tree has an infinite path.
\end{definition}

Note that in the context of constructive reverse mathematics (e.g. in~\cite{dienerConstructiveReverseMathematics2020}) $\WKL$ is only stated for decidable trees.
We however need both notions and thus distinguish them by an index.

$\WKL$ is a consequence of compactness for decidable models.
The proof is essentially the same as the one for propositional logic and $\WKL_{\mathcal{D}}$ by Diener~\cite{dienerConstructiveReverseMathematics2020} and the one for first-order logic using the classical base theory $\mathsf{RCA}_0$ by Simpson~\cite{simpson2009subsystems}.
\newcommand{\filter}[3]{#2 \subseteq_{#1} #3}
Intuitively, given a tree $\tau$, one can construct a formula $\phi_n$ over the siganture $\Sigma_{\Nat}$ which is satisfiable iff $\tau$ contains an element of length $n$.

\begin{definition}[][count_sig]
  We define the signature $\Sigma_{\Nat} := (\bot, \Nat)$ with constant arity $0$, i.e. no term symbols and countably many propositions $P_i$.
\end{definition}

\begin{fact}[][listable_list_length]
  There is a function $L_{-} : \Nat \to \List(\List (\Bool))$ s.t. $\forall l : \List (\Bool).\, |l| = n \leftrightarrow l \in L_n$.
\end{fact}

For example the tree $\tau l := l = [] \lor \exists l'.~l = \btrue :: l'$ contains the elements $$[\btrue, \btrue, \btrue], [\btrue, \btrue, \bfalse],[\btrue, \bfalse, \btrue], [\btrue, \bfalse, \bfalse]$$ of length $3$ and we build the formula
$$\phi_3 := (P_0 \dot\land P_1 \dot\land P_2) \dot\lor (P_0 \dot\land P_1 \dot\land \dot\neg P_1) \dot\lor (P_0 \dot\land \dot\neg P_1 \dot\land P_2) \dot\lor (P_0 \dot\land \dot\neg P_1 \dot\land \dot\neg P_1).$$

Now for an infinite tree every context in the theory $\TT := \{ \phi_n \mid n : \Nat \}$ is satisfiable by an (omniscient) model, because the tree is infinite.
By compactness, the interpretation of $P_i$ in a decidable model for the whole theory yields an infinite path through $\tau$.

Since trees are not necessarily decidable predicates, it is not possible to construct a list of elements up to length $n$ explicitly, and thus not possible to construct $\phi_n$ explicitly.
However, we can prove the double-negation of the existence of such lists using a filtering predicate:

\begin{definition}[][Is_filter]
We define the filtering $L'$ of a list $L$ under a predicate $P$:
\[\infer{\filter P {[]}{[]}}{~}\hspace{3em}
  \infer{\filter P {x :: L} {x :: L'}}{P x \quad \filter P L L'} \hspace{3em}
  \infer{\filter P {x :: L} {L'}}{\neg P x \quad \filter P L L'}
\]
\end{definition}

\begin{fact}[][Is_filter_exists]
$\forall L.\neg\neg \exists L'. \filter P L {L'}$
\end{fact}

Given a tree $\tau$ we can define the (singleton) theory $\TT_n$ where all elements have the shape of $\phi_n$ described above:
$$\TT_n := \left\{\phi \mid \exists L. \filter \tree {L_n} L \land \phi = \bigvee \left[ \bigwedge \left[ P_i^{(b)} \mid b \in l \textit{ at position } i \right]\mid l \in L \right]\right\}$$ where $P_i^{(b)} := P_i $ if $b = \btrue$ and $P_i^{(b)} := \dot\neg P_i$  if $b = \bfalse$.

This preparation now suffices to prove the following central Lemma:

\begin{lemma}[][compact_implies_WKL]
  Given a tree $\tree$ one can construct a theory $\mathcal T$ over $\Sigma_{\Nat}$ s.t.
  \begin{enumerate}
    \coqitem[infinite_finitely_satisfiable] If $\tree$ is infinite, all $\Gamma \subseteq \mathcal T$ have an omniscient model.
    \coqitem[decidable_to_decidable] If $\tree$ is infinite and decidable, $\mathcal T$ is decidable.
    \coqitem[exists_quasi_path] If $\mathcal T$ has a decidable model, $\exists f : \Nat \to \Bool.\forall n. \neg\neg \tau [f\,0,\dots,f\,n]$.
  \end{enumerate}
\end{lemma}
\begin{proofqed}
  Let $\tree$ be given.
  We define the theory $\TT := \{\phi \mid \exists n.\,P\,n\,\phi\}$

  Given $u : \List(\Bool)$ we can define the omniscient model $\MM_u$ which satisfies the atom $a_i$ if $i > |u|$ or the $i$-th element of $u$ is $\btrue$.

  We need the following lemmas:
  \begin{enumerate}[(a)]
    \coqitem[get_index_list] $\forall \Gamma \subseteq \TT.\exists L : \List(\Nat).\forall \phi \in \Gamma.\exists n \in L. \phi \in \TT_n$. 
    \coqitem[phi_down] For omniscient models $\MM$ and $\phi_n \in \TT_n$, $\phi_m \in \TT_m$, and $n \geq m$ we have $\rho \vDash_{\MM} \phi_n \to \rho \vDash_{\MM} \phi_m$.
    \coqitem[model_u] For all $u$ s.t. $\tree u$ we have $\forall n \leq |u|.\forall \phi.\, \phi \in \TT_n \to \MM_u \vDash \psi$.
    \coqitem[phi_exists] $\forall n. \neg\neg\exists \phi.\, \phi \in \TT_n$.
    \coqitem[exists_quasi_path] If $\TT$ has a decidable model $\MM$, the function $f i := \btrue$ if $\MM \vDash P_i$ and $f i := \bfalse$ otherwise fulfills $\forall n \phi.\, \TT_n\, \phi \to T [f\,0,\dots,f\,n]$.
  \end{enumerate}

  The proof of (a) is by induction on $\Gamma$.
  The proof of (b) is technical but not hard.
  The proof of (c) is straightforward using (b).
  The proof of (d) uses \Cref{coq:Is_filter_exists}.
  The proof of (e) is again technical but not hard.
  
  Now for claim (1) let $\tree$ be infinite and $\Gamma \subseteq \TT$.
  We use (a) and compute the maximum $m$ of $L$.
  By infinity of $\tree$ there is $u$ s.t.\ $\tree u$ and $|u| \geq m$ and by (c) $\MM_u$ satisfies all $\phi \in \Gamma$.

  Claim (2) is by computing the filtering of $L_n$ using the decider for $\tree$.

  Claim (3) is immediate from (d) and (e).
\end{proofqed}

\begin{corollary}[][compact_implies_WKL]
  Given $\EM$, compactness for decidable models implies $\WKL$.
\end{corollary}
\begin{corollary}[][compact_implies_WKL_D]
  Compactness for decidable models implies $\WKL_{\mathcal D}$, even if compacteness is only assumed for decidable theories.
\end{corollary}

Note that since compactness for enumerable theories implies compactness for decidable theories, the latter implication also holds for compactness w.r.t.\ enumerable theories.
It seems however that this proof cannot be directly strengthened to also yield $\WKL$ for enumerable trees.

As a last step we prove that $\WKL$ makes every logically decidable predicate on data types decidable:

\begin{lemma}[][WKL_to_decidable]
  Let $p$ s.t. $\forall n : \Nat.\, p\, n \lor \neg p\, n$.
  Then $\WKL$ implies that $p$ is decidable.
\end{lemma}
\begin{proofqed}
  We define a tree $\tree$ which contains prefixes of a decider for $p$ as $\tree\, u := \forall b \in u \text{ at position } i.\, b = \btrue \leftrightarrow p\, i$.
  Now $\tree$ is infinite because we can prove the \emph{existence} of lists of arbitrary length in $\tree$.
  An infinite path through $\tau$ decides $p$.
\end{proofqed}
\begin{corollary}[][CO_iff_EM_WKL]
  $\EM \land \WKL$ implies that every predicate on data types is decidable.
\end{corollary}

\begin{lemma}[][WKL_implies_modex]\label{lem:WKL_to_modex}
  If every predicate on data types is decidable, model existence for omniscient models holds.
\end{lemma}
\begin{proofqed}
  Model existence is constructively provable.
  The model is omniscient by assumption since $\vDash_{\MM}$ can equivalently be seen as a predicate $\List(D) \times \Form \to \Prop$.
\end{proofqed}

This suffices to state our final equivalence theorem for Tarski semantics:

\setCoqFilename{Analysis}
\begin{theorem}\label{tarski_theorem}
  The following are equivalent for arbitrary theories:%
  \begin{enumerate}
  \coqitem[iff_1a_2a] Completeness of $\TT \vdash_c \phi$ for omniscient/decidable models.
  \coqitem[impl_5_2a] $\EM$ and model existence for omniscient/decidable models.
  \coqitem[impl_2a_3a] $\EM$ and compactness for omniscient/decidable models.
  \coqitem[impl_3a_4] $\EM$ and $\WKL$.
  \coqitem[impl_4_5] Every predicate on data types is decidable.
  \end{enumerate}
\end{theorem}

In other words, completeness of $\TT \vdash_c \phi$ w.r.t.\ omniscient and decidable models is equivalent to a purely logical principle ($\EM$) and a function existence principle ($\WKL$).
Item (5) is reminiscent of the axiom $\mathsf{CC}^\lor$ used in~\cite{berger2012weak}, where $\WKL_{\mathcal D}$ is proven equivalent to $\mathsf{LLPO} \land \Pi^0_1\textsf{-CC}^\lor$
Given that under $\EM$, $\WKL$ is equivalent to Brouwer's fan theorem $\mathsf{FAN}$~\cite{berger2012weak}, we could have stated item (4) equivalently as $\EM \land \mathsf{FAN}$.

If one is interested in similar equivalences to completeness of $\Gamma \vdash_c \phi$ and $\TT \vdash_c \phi$ for enumerable $\TT$ w.r.t.\ omniscient and decidable models, our previous analysis has yielded that the corresponding logical principles will be $\MPL$ and $\MP$ respectively instead of $\EM$.
To the best of our knowledge, it is an open question which function existence or choice principle replaces $\WKL$ in this situation, i.e. which adjustment of $\WKL$, $\textsf{CC}^\lor$, or $\mathsf{FAN}$ becomes necessary.

\subsection{Kripke Semantics}
\label{sec:kripke}
\setCoqFilename{Kripke}

Turning to intuitionistic logic, we present Kripke semantics immediately generalised to arbitrary interpretations of falsity.

\begin{definition}[][kmodel]
	A \emph{Kripke model} \KK over a domain $D$ is a preorder $(\WW,\preceq)$ with
	$$\_^\KK~:~\forall f:\Funcs.\,D^{|f|}\to D\hspace{2em}
	\_^\KK~:~\forall P:\Preds.\,\WW\to D^{|P|}\to \Prop\hspace{2em}
	\bot^\KK~:~\WW\to\Prop.$$
	The interpretations of predicates and falsity are required to be monotone, i.e. $P_v^\KK\,\vec a\to P_w^\KK\,\vec a$ and $\bot_v^\KK\to\bot_w^\KK$ whenever $v\preceq w$.
	Assignments $\rho$ and their term evaluations $\hat\rho$ are extended to formulas via the relation $w\Vdash_\rho \phi$ defined by
	\begin{align*}
	w\Vdash_\rho \dot{\bot}&~:=~\bot^\KK_w&
	w\Vdash_\rho \phi\dot\to\psi &~:=~\forall v\succeq w.\,v\Vdash_\rho\phi\to v\Vdash_\rho\psi\\
	w\Vdash_\rho P\,\vec t\,&~:=~P_w^\KK\,(\hat{\rho}\,@\,\vec t\,)&
	w\Vdash_\rho\dot{\forall}\,\phi&~:=~\forall a:D.\,w\Vdash_{a;\rho} \phi
	\end{align*}
	We write $\KK\Vdash\phi$ if $w\Vdash_\rho \phi$ for all $\rho$ and $w$.
	\KK is \emph{standard} if $\bot^\KK_w$ implies $\bot$ for all $w$ and \emph{exploding} if $\KK\Vdash \dot{\bot}\dot{\to}\phi$ for all $\phi$.
	We write $\TT\Vdash \phi$ if $\KK\Vdash_\rho \phi$ for all standard $\KK$ and $\rho$ with $\KK\Vdash_\rho \TT$, and $\TT\Vdash_e \phi$ when relaxing to exploding models.
\end{definition}

Note that standard models are exploding, hence $\TT\Vdash_e \phi$ implies $\TT\Vdash \phi$.
Moreover, the monotonicity required for the predicate and falsity interpretations lifts to all formulas, i.e. $w\Vdash_\rho \phi$ implies $v\Vdash_\rho \phi$ whenever $w\preceq v$.
This property together with the usual facts about the interaction of assignments and substitutions yields soundness:

\begin{fact}[][ksoundness']
	\label{fact:ksoundness}
	$\Gamma\vdash\phi$ implies $\Gamma\Vdash_e \phi$.
\end{fact}

\begin{proofqed}
	By induction on $\Gamma\vdash\phi$ and analogous to~\cite[Fact 3.34]{ForsterCPP}.
\end{proofqed}

Turning to completeness, instead of showing that $\Gamma\Vdash_e \phi$ implies $\Gamma\vdash\phi$ directly, we follow Herbelin and Lee~\cite{HerbelinCut} and reconstruct a formal derivation in the normal sequent calculus LJT, hence implementing a cut-elimination procedure.
LJT is defined by judgements $\Gamma\sq \phi$ and $\Gamma\,;\psi\sq \phi$ for a focused formula $\psi$:
\begin{gather*}
	\hspace{1em}
	\infer[\small\textnormal{A}]{\Gamma \,; \varphi \sq \varphi}{}\hspace{3em}
	\infer[\small\textnormal C]{\Gamma \sq \psi}{\Gamma \,; \varphi \sq \psi \quad \varphi \in \Gamma}\hspace{3em}
	\infer[\small\textnormal{IL}]{\Gamma \,; \varphi \dot\to \psi\sq \theta}{\Gamma \sq \varphi \quad \Gamma \,; \psi \sq \theta}\\[0.3cm]
	\infer[\small\textnormal{IR}]{\Gamma \sq \varphi \dot\to \psi}{\Gamma,\varphi \sq \psi}\hspace{3em}
	\infer[\small\textnormal{AL}]{\Gamma \,; \dot\forall \varphi \sq \psi}{\Gamma \,; \varphi[t]\sq \psi}\hspace{3em}
	\infer[\small\textnormal{AR}]{\Gamma \sq \dot\forall \varphi}{\up{\Gamma} \sq \varphi}\hspace{3em}
	\infer[\small\textnormal E]{\Gamma \sq \varphi}{\Gamma \sq \dot\bot}
\end{gather*}

\setCoqFilename{Gentzen}
\begin{fact}[][cutfree_seq_ND]
	\label{fact:cutfree}
	Every sequent $\Gamma\sq\phi$ can be translated into a normal derivation $\Gamma\vdash\phi$.
\end{fact}

\begin{proofqed}
	By simultaneous induction on both forms of judgements, where every sequent $\Gamma\,;\psi\sq \phi$ is translated to an implication from $\Gamma\vdash \psi$ to $\Gamma\vdash  \phi$.
\end{proofqed}

By the previous fact, completeness for LJT implies completeness for intuitionistic ND.
The technique to establish completeness for Kripke semantics is based on universal models coinciding with intuitionistic provability.
We in fact construct two syntactic Kripke models over the domain $\Term$.

\begin{itemize}
	\item
	An exploding model $\UU$ on contexts s.t. $\Gamma\Vdash_\sigma^\UU\phi $ iff $\Gamma\sq \phi[\sigma]$.
	\item
	A standard model $\CC$ on consistent contexts s.t. $\Gamma\Vdash^\CC_\sigma\phi$ iff $\neg\neg(\Gamma\sq \phi[\sigma])$.
\end{itemize}

\setCoqFilename{KripkeCompleteness}

These constructions are
adaptions of those in \cite{Wehr2019}, which in turn are based on the
proof and comments in \cite{HerbelinCut}.
We begin with the exploding model $\UU$. 

\begin{definition}[][K_ctx]
	The model $\UU$ over the domain $\Term$ of terms is defined on the
	contexts $\Gamma$ preordered by inclusion $\subseteq$. Further, we
	set:
	$$f^{\,\UU} \,\vec{d} := f\,\vec{d} \qquad \qquad P^{\,\UU}_{\,\Gamma}\,\vec{d} := \Gamma \sq
	P\,\vec{d} \qquad \qquad \bot^{\,\UU}_\Gamma := \Gamma \sq \dot\bot $$
\end{definition}

The desired properties of $\UU$ can be derived from
the next lemma, which takes the shape of a normalisation-by-evaluation
procedure~\cite{NBEBerge,NBE}.

\begin{lemma}[][K_ctx_correct]
	\label{thm:um-corr}
	In the universal Kripke model $\UU$ the following hold.
	\begin{enumerate}
		\item $\Gamma \Vdash_\sigma \varphi \to \Gamma \sq \subst{\varphi}{\sigma}$
		\item $(\forall\, \Gamma' \psi.~\Gamma \subseteq \Gamma' \to \Gamma' \,;
		\subst{\varphi}{\sigma} \sq \psi \to \Gamma' \sq \psi) \to \Gamma \Vdash_\sigma
		\varphi$
	\end{enumerate}
\end{lemma}
\begin{proof}
	We prove (1) and (2) at once by induction on $\varphi$ generalising $\Gamma$ and $\sigma$.
	We only discuss the case of implications $\varphi \dot\to \psi$ in full detail.
	
	\begin{enumerate}
		\item Assuming $\,\forall \Gamma'.~\Gamma \subseteq \Gamma' \to \Gamma' \Vdash_\sigma \varphi \to
		\Gamma' \Vdash_\sigma \psi$, one has to derive that \mbox{$\Gamma \sq \subst{(\varphi \dot\to \psi)}{\sigma}$}.
		Per (IR) and inductive hypothesis (2) for
		$\psi$ it suffices to show $\cextend{\Gamma}{\subst{\varphi}{\sigma}} \Vdash_\sigma
		\psi$. Applying the inductive hypothesis (2) for $\varphi$ and the assumption,
		it suffices to show that $\Gamma' \,; \subst{\varphi}{\sigma} \sq
		\subst{\theta}{\sigma}$ implies $\Gamma' \sq \subst{\theta}{\sigma}$ for any
		$\cextend{\Gamma}{\subst{\varphi}{\sigma}} \subseteq \Gamma'$ and $\theta$, which holds per (C).
		\item Assuming $\forall \,\Gamma'\,\theta.~\Gamma \subseteq \Gamma' \to
		\Gamma' \,;\subst{(\varphi \dot\to \psi)}{\sigma} \sq \theta \to \Gamma' \sq \theta$ one
		has to deduce $\Gamma' \Vdash_\sigma \varphi$ entailing $\Gamma'
		\Vdash_\sigma \psi$ for any $\Gamma \subseteq \Gamma'$. Because of the inductive hypothesis (2) for $\psi$ it suffices
		to show $\Delta \,;\subst{\psi}{\sigma} \sq \theta$ implying $\Delta
		\sq \theta$ for any $\Gamma' \subseteq \Delta$.
		By using the
		assumption, $\Delta \sq \theta$ reduces to $\,\Delta \,;
		\subst{(\varphi \dot\to \psi)}{\sigma} \sq \theta$. This follows by
		(IL), as the assumption $\Gamma' \Vdash_\sigma \varphi$ implies
		$\Delta \sq \subst{\varphi}{\sigma}$ per inductive hypothesis (2).
		\qed
	\end{enumerate}
	
\end{proof}

\begin{corollary}[][K_ctx_constraint]
	\label{thm:univ-iff}
	$\UU$ is exploding and satisfies $\Gamma \Vdash_\sigma \varphi$ iff $\Gamma \sq \subst{\varphi}{\sigma}$.
\end{corollary}
\begin{proofqed}
	Suppose that $\Gamma\sq\dot\bot$, then (2) of \Cref{thm:um-corr} yields that $\Gamma\Vdash_\sigma\phi$ for arbitrary $\phi$.
	Thus $\UU$ is exploding.
	The claimed equivalence then follows by (1) of \Cref{thm:um-corr} and
	soundness of LJT.
\end{proofqed}

Being universal, $\UU$ witnesses completeness for exploding Kripke models:

\begin{fact}[][K_exp_completeness]
	\label{fact:kexploding}
	\begin{enumerate}
		\item
		$\Gamma\Vdash_e \phi$ implies $\Gamma\sq\phi$.
		\item
		In the $\to,\forall$-fragment, $\Gamma\Vdash \phi$ implies $\Gamma\sq\phi$.
	\end{enumerate}
\end{fact}

\begin{proof}
	\begin{enumerate}
		\item
		Since $\Gamma\Vdash_\id^\UU\Gamma$ we have that $\Gamma\Vdash_e \phi$ implies $\Gamma\Vdash_\id^\UU\phi $ and hence $\Gamma\sq \phi$.
		\item
		In the minimal fragment, $\dot{\bot}$ remains uninterpreted and hence imposes no condition on the models.
		Hence $\UU$ yields the completeness in this case.
	\end{enumerate}	
\end{proof}

Before we move on to completeness for standard models, we illustrate how the previous fact already establishes the cut rule for LJT.

\begin{lemma}[][SE_cut]
	\label{lem:cut}
	If $\Gamma\sq \phi$ and $\Gamma;\phi\sq \psi$, then $\Gamma\sq \psi$.
\end{lemma}

\begin{proof}
	By the translation given in \Cref{fact:cutfree}, we obtain a derivation $\Gamma\vdash \psi$ from the two assumptions.
	This can be turned into $\Gamma\sq \psi$ using soundness (\Cref{fact:ksoundness}) and completeness (\Cref{fact:kexploding}).
\end{proof}

We now construct the universal standard model $\CC$ as a refinement of $\UU$.
As standard models require that $\bot^\KK_v$ implies $\bot$ for any $v$,
the model $\UU$ has to be restricted to the consistent contexts, those which do
not prove $\dot\bot$.

\begin{definition}[][K_std]
	The model $\CC$ over the domain $\Term$ of terms is defined on the
	consistent contexts $\Gamma \not\sq \dot\bot$ preordered by inclusion $\subseteq$. Further, we
	set:
	$$f^{\,\CC} \,\vec{d} := f\,\vec{d} \qquad \qquad P^{\,\CC}_{\,\Gamma}\,\vec{d} := \neg\neg(\Gamma \sq
	P\,\vec{d}) \qquad \qquad \bot^{\,\CC}_{\Gamma} := \bot $$
\end{definition}

Note that $\CC$ is obviously standard and that we weakened the interpretation of atoms to doubly negated provability.
This admits the following normalisation-by-evaluation procedure for doubly negated sequents:

\begin{lemma}[][K_std_correct]
	\label{thm:cm-corr}
	In the universal Kripke model $\CC$ the following hold.
	\begin{enumerate}
		\item $\Gamma \Vdash_\sigma \varphi \to \neg\neg( \Gamma \sq \subst{\varphi}{\sigma})$
		\item $(\forall\, \Gamma' \psi.~\Gamma \subseteq \Gamma' \to \Gamma' \,;
		\subst{\varphi}{\sigma} \sq \psi \to \neg\neg(\Gamma' \sq \psi)) \to \Gamma \Vdash_\sigma
		\varphi$
	\end{enumerate}
\end{lemma}
\begin{proof}
	We prove (1) and (2) at once by induction on $\varphi$ generalising $\Gamma$ and $\sigma$.
	Most cases are completely analogous to those in \Cref{thm:um-corr}.
	Therefore we only discuss the crucial case (1) for implications $\varphi \dot\to \psi$.

	\begin{enumerate}
		\item
		Assuming $\Gamma \Vdash_\sigma \varphi\dot{\to}\psi$ we need to derive \mbox{$\neg\neg(\Gamma \sq \phi[\sigma] \dot\to \psi[\sigma])$}.
		So we assume $\neg(\Gamma \sq \phi[\sigma] \dot\to \psi[\sigma])$ and derive a contradiction.
		Because of the negative goal, we may assume that either $\Gamma,\phi[\sigma]$ is consistent or not.
		In the positive case, we proceed as in \Cref{thm:um-corr} since the extended context is a node in $\CC$.
		On the other hand, if $\Gamma,\phi[\sigma]\sq\dot{\bot}$, then $\Gamma,\phi[\sigma]\sq\psi[\sigma]$ by (E) and hence 
		$\Gamma \sq \phi[\sigma] \dot\to \psi[\sigma]$ by (IR), contradicting the assumption.
		\qed
	\end{enumerate}
\end{proof}

\begin{corollary}[][K_std_sprv]
	$\CC$ satisfies $\Gamma \Vdash_\sigma \varphi$ iff $\neg\neg(\Gamma \sq
	\subst{\varphi}{\sigma})$.
\end{corollary}
\begin{proofqed}
	The first direction is (1) of \Cref{thm:cm-corr} and the converse follows with (2) since $\neg\neg (\Gamma \sq \phi[\sigma])$ and $\Gamma' \,;
	\subst{\varphi}{\sigma} \sq \psi$ for $\Gamma'\supseteq \Gamma$ together imply $\neg\neg (\Gamma' \sq \psi)$ via the cut rule established in \Cref{lem:cut}.
\end{proofqed}

The advantage of the additional double negations is that, in contrast to the proof in~\cite{HerbelinCut}, we only need a single application of stability to derive completeness.
Thus we can prove the completeness of $\Gamma\vdash\phi$ admissible in~\Cref{sec:L}.

\begin{fact}[][K_std_completeness]
	\begin{enumerate}
		\item
		$\Gamma\Vdash \phi$ implies $\Gamma\sq\phi$, provided that $\Gamma\sq\phi$ is stable.
		\item
		$\Gamma\Vdash \phi$ implies $\Gamma\vdash\phi$, provided that $\Gamma\vdash\phi$ is stable.
	\end{enumerate}
\end{fact}

\begin{proof}
	\begin{enumerate}
		\item
		Since $\Gamma\Vdash \phi$ implies $\neg\neg(\Gamma \sq
		\varphi)$, we can conclude $\Gamma\sq\phi$ per stability.
		\item
		Since $\Gamma\sq \phi$ iff $\Gamma\vdash \phi$ per soundness and completeness (Facts~\ref{fact:ksoundness} and \ref{fact:kexploding}).
		\qed
	\end{enumerate}
\end{proof}

Conversely, unrestricted completeness requires the stability of classical ND.

\begin{fact}[][cend_dn]
	Completeness of $\Gamma\sq\phi$ implies stability of $\Gamma\vdash_c\phi$.
\end{fact}

\begin{proofqed}
	Assume completeness of $\Gamma\sq\phi$ and suppose $\neg\neg(\Gamma\vdash_c\phi)$.
	We prove $\Gamma\vdash_c\phi$, so it suffices to show $\Gamma,\dot\neg\phi\vdash_c\dot\bot$.
	Employing a standard double-negation translation $\phi^N$ on formulas $\phi$, it is equivalent to establish $(\Gamma,\dot\neg\phi)^N\sq\dot\bot$.
	Applying completeness, however, we may assume a standard model $\KK$ with $\KK\Vdash_\rho (\Gamma,\dot\neg\phi)^N$ and derive a contradiction.
	Hence we conclude $\Gamma\vdash_c\phi$ and so $\Gamma^N\Vdash \phi^N$ from $\neg\neg(\Gamma\vdash_c\phi)$ and soundness, in conflict to $\KK\Vdash_\rho (\Gamma,\dot\neg\phi)^N$.
\end{proofqed}

Thus, the completeness of intuitionistic ND is similar to the classical case.

\setCoqFilename{Analysis}
\begin{theorem}
	\label{thm:icharacterisations}
	\begin{enumerate}
		\coqitem[kcompleteness_iff_MPL]
		Completeness of $\Gamma\vdash\phi$ is equivalent to $\MPL$.
		\coqitem[kcompleteness_enum_implies_MP]
		Completeness of $\TT\vdash\phi$ for enumerable \TT implies \MP.
		\coqitem[kcompleteness_implies_XM]
		Completeness of $\TT\vdash\phi$ for arbitrary \TT implies \EM.
	\end{enumerate}
\end{theorem}
\setCoqFilename{KripkeCompleteness}

\subsection{On Markov's Principle}
\label{sec:L}
\setCoqFilename{Markov}

We show that the stability of $\Gamma \vdash_c \phi$ and $\Gamma \vdash \phi$ is equivalent to an object-level version of Markov's Principle referencing procedures in a concrete model of computation.
For mechanisation purposes, we will use the call-by-value $\lambda$-calculus $\L$~\cite{Plotkin75,ForsterL} as model of computation.
Since on paper the same proofs can be carried out for any model of computation we will not go into details of $\L$.
We only need two notions:
first, $\L$-enumerability~\cite[Definition 6]{ForsterLambda}, which is defined like synthetic enumerability, but where the enumerator is an $\L$-computable function.
Secondly, the halting problem for $\L$, defined as  $\eva s := \textit{``the term $s$ terminates''}$.

We define the object-level Markov's Principle $\MPL$ as stability of $\eva$:
$$\MPL := \forall s.~\neg\neg\eva{s} \to \eva{s}$$

$\MPL$ can also be phrased similarly to $\MP$ with a condition on the sequence:

\begin{lemma}[{\cite[Theorem 45]{ForsterL}}]
	$\MPL$ is equivalent to \[\forall f:\Nat\to\Bool.\,\text{$\L$-computable}\,f \to \neg\neg(\exists n.\,f\,n=\btrue)\to \exists n.\,f\,n=\btrue.\]
\end{lemma}
\begin{corollary}[][MP_MPL]
	$\MP$ implies $\MPL$.
\end{corollary}

We show \Cref{lem:MPL-equivs}, i.e.~that $\MPL$ is equivalent to both the stability of $\vdash_c$ and $\vdash$ for contexts, thereby establishing that completeness of provability for standard Tarski and Kripke semantics for finite theories is equivalent to $\MPL$.

\begin{lemma}[{\cite[Fact 2.16]{ForsterCPP}}]
	Let $p$ and $q$ be predicates.
	If~$p$~many-one reduces to $q$ (i.e.~$\exists f.\forall x.\,p x \leftrightarrow q (f x)$, written $p \preceq q$) and $q$ is stable, then $p$ is stable.
\end{lemma}

Thus, in order to prove the equivalence of the stability of $\eva$, $\Gamma \vdash \phi$, and $\Gamma \vdash_c \phi$, it suffices to give many-one reductions between them.
We start with the two simpler reductions:

\begin{lemma}[][cprv_iprv]
	$\vdash_c\, \preceq~\vdash$, and thus stability of $\Gamma\vdash\phi$ implies the stability of $\Gamma\vdash_c\phi$.
\end{lemma}
\begin{proofqed}
	Using a standard double-negation translation proof.
\end{proofqed}

\begin{lemma}[][halt_cprv]
	$\eva \preceq ~\vdash_c$, and thus stability of $\Gamma\vdash_c\phi$ implies $\MPL$.
\end{lemma}
\begin{proofqed}
	$\eva$ reduces to the halting problem of multi-tape Turing machines~\cite{Wuttke2018}, which reduces to the halting problem of single-tape Turing machines~\cite{forster2019TMs}, which reduces to the Post correspondence problem~\cite{forster2018verification}, which in turn reduces to $\vdash_c$ by adapting~\cite[Corollary 3.49]{ForsterCPP}.
\end{proofqed}

Since $p \preceq \eva$ for all $\L$-enumerable predicates $p$~\cite[Theorem 7]{ForsterLambda}, it suffices to give an $\L$-computable enumeration of type $\Nat \to \List (\Form)$ of provable formulas $\vdash \phi$.
Note that we continue to assume signatures to be (synthetically) enumerable and do \emph{not} have to restrict to $\L$-enumerability, which is enabled by the following signature extension lemma:

\setCoqFilename{Extend}
\begin{lemma}[][prv_embed]
	\label{coq:extension}
	Let $\iota$ be an invertible embedding from $\Sigma$ to $\Sigma'$.
	Then $\vdash \phi$ over $\Sigma$ if and only if $\vdash \iota \phi$ over $\Sigma'$, where $\iota \phi$ is the recursive application of $\iota$ to formulas.
\end{lemma}
\begin{proofqed}
	$\Gamma \vdash \phi \to \iota \Gamma \vdash \iota \phi$ follows trivially by induction.
	For the inverse direction, we show that Kripke models $M$ over $\Sigma$ can be extended to Kripke models $\iota M$ over $\Sigma$ s.t. $\rho, u \Vdash_M \phi \leftrightarrow \rho, u \Vdash_{\iota M} \iota \phi$.
	Then $\iota \Gamma \vdash \iota \phi \to \Gamma \vdash \phi$ follows from soundness and completness w.r.t. exploding models.
\end{proofqed}

\setCoqFilename{LEnum}
\begin{lemma}[][enum_sprvie]
	$\Gamma\vdash\phi$ is $\L$-enumerable for any enumerable signature $\Sigma$.
\end{lemma}
\begin{proofqed}
	Since $\Sigma$ is enumerable, it can be injectively embedded via $\iota$ into the maximal signature $\Sigma_{\text{max}} := (\Nat^2, \Nat^2)$ where the arity functions are just the second projections.
	Since $\Nat^2$ is $\L$-enumerable, terms and formulas over $\Sigma_{\text{max}}$ are also $\L$-enumerable, and thus provability over $\Sigma_{\text{max}}$ is $\L$-enumerable.
	By~\Cref{coq:extension} we obtain that provability over $\Sigma$ is $\L$-enumerable.
\end{proofqed}

\setCoqFilename{Markov}
\begin{corollary}[][iprv_halt]
	$\vdash\, \preceq \eva$, and thus $\MPL$ implies the stability of $\Gamma\vdash\phi$.
\end{corollary}

We conclude the section with observations on independence and admissibility of several statements in Coq's type theory pCuIC.
By \emph{independence} of a statement $P$, we mean that neither $P$ nor $\neg P$ is provable in pCuIC without assumptions.
By \emph{admissibility} of a statement $\forall x.~P(x) \to Q(x)$ we mean that whenever $P(t)$ is provable in pCuIC for a concrete term $t$ without assumptions, $Q(t)$ is as well.
Pédrot and Tabareau~\cite{PedrotMP} show $\MP$ independent (Corollary 41) and admissible (Theorem 33).
This transports to $\MPL$ as well as stability of deduction systems and completeness with respect to model-theoretic semantics.

\newtheorem{observation}[theorem]{Observation}

\begin{theorem}[][MPL_independent]
  The following are all independent and admissible in pCuIC:
  \begin{enumerate}
  \item $\MPL$ %
  \item Stability of both $\Gamma\vdash_c\phi$ and $\Gamma\vdash\phi$.
  \item Completeness of $\mathcal{T}\vdash_c\phi$ for enumerable $\mathcal{T}$ w.r.t.~standard Tarski semantics.
  \item Completeness of $\Gamma\vdash_c\phi$ w.r.t.~standard Tarski semantics.
  \item Completeness of $\Gamma\vdash \phi$ w.r.t.~standard Kripke semantics.
  \end{enumerate}
\end{theorem}
\begin{proofqed}
  We exemplarily show (1) and (4), the other proofs are similar.

  For (1), $\MPL$ is consistent since it is a consequence of $\mathsf{EM}$.
  Lemma 40 in~\cite{PedrotMP} shows that no theory conservative over the calculus of inductive constructions (CIC) can prove both the independence of premise rule $\mathsf{IP}$ and $\MP$, by turning these assumptions into a decider for the halting problem of the untyped term language of CIC.
  One can adapt the proof to show that pCuIC cannot prove both $\mathsf{IP}$ and $\MPL$, by constructing a decider for the $\L$-halting problem instead, which yields a contradiction as well.
  The admissibility of $\MPL$ follows from the admissibility of $\MP$ since a single application of $\MP$ suffices to derive $\MPL$.

  For (4), independence follows directly from (1) and \Cref{thm:characterisations}.
  For admissibility, assume that $\Gamma \vDash \phi$ is provable in pCuIC.
  By \Cref{thm:quasi}, $\neg\neg(\Gamma \vdash_c \phi)$ is provable in pCuIC.
  Thus by (2), $\Gamma \vdash_c \phi$ is provable in pCuIC.
\end{proofqed}

\section{Algebraic Semantics}
\label{sec:algebras}
\setCoqFilename{Heyting}

In contrast to the model-theoretic semantics discussed in \Cref{sec:models}, algebraic semantics are not based on models interpreting the non-logical symbols but on algebras suitable for interpreting the logical connectives of the syntax.
A formula is valid if it is satisfied by all algebras and completeness follows from the observation that deduction systems have the corresponding algebraic structure.
Following \cite{scott_algebraic_2008}, we discuss complete Heyting and Boolean algebras coinciding with intuitionistic and classical ND, respectively.
We consider all formulas $\phi:\Form$.

\begin{definition}[][HeytingAlgebra]
	\label{def:heyting}
	A \emph{Heyting algebra}  consists of a preorder $(\H,\le)$ and operations
	$$0:\H,\hspace{1cm}\sqcap: \H\to \H\to \H,\hspace{1cm}\sqcup: \H\to \H\to \H,\hspace{1cm}\Rightarrow:	 \H\to \H\to \H$$
	for bottom, meet, join, and implication satisfying the following properties:
	\vspace{-0.3cm}
	\begin{multicols}{2}
	\begin{enumerate}
		\item
		$0\le x$
		\item
		$z\sqcap x\le y\leftrightarrow z \le x\Rightarrow y$
		\item
		$z\le x\land z\le y\leftrightarrow z \le x\sqcap y$
		\item
		$x\le z\land y\le z\leftrightarrow x\sqcup y \le z$
	\end{enumerate}	
	\end{multicols}
	Moreover, $\mathcal H$ is \emph{complete} if there is an operation $\bigsqcap:(\H\to \Prop)\to \H$ for arbitrary meets satisfying $(\forall y\in P.\, x\le y)\leftrightarrow x\le \bigsqcap P$.
	Then $\mathcal \H$ also has arbitrary joins $\bigsqcup P:= \bigsqcap (\lambda x.\,\forall y\in P.\, y\le x)$ satisfying $(\forall y\in P.\, y\le x)\leftrightarrow  \bigsqcup P\le x$.
\end{definition}

Arbitrary meets and joins indexed by a function $F:I\to \H$ on a type $I$ are defined by $\bigsqcap_i F\,i:=\bigsqcap (\lambda x.\,\exists i.\,x=F\,i)$ and $\bigsqcup_i F\,i:=\bigsqcup (\lambda x.\,\exists i.\,x=F\,i)$, respectively.
As we do not require $\le$ to be antisymmetric in order to avoid quotient constructions, we establish equational facts about Heyting algebras only up to equivalence $x\equiv y:=x\le y\land y\le x$ rather than actual equality.

\begin{lemma}
	\label{lem:distr}
	Let \H be a Heyting algebra.
	\begin{enumerate}
		\coqitem[meet_join_distr]
		\H is $\sqcap$-$\sqcup$-\emph{distributive}, i.e. $x\sqcap (y\sqcup z)\equiv (x\sqcap y)\sqcup (x\sqcap z)$.
		As a consequence, $x\le y \sqcup z$ implies $x\le (x\sqcap y)\sqcup (x\sqcap z)$.
		\coqitem[meet_sup_distr]
		If \H is complete then it is $\sqcap$-$\bigsqcup$-\emph{distributive}, i.e. $x\sqcap (\bigsqcup_i F)\equiv\bigsqcup_i (\lambda i.\,x\sqcap F\,i)$.
		As a consequence, $x\le \bigsqcup_i F$ implies $x\le \bigsqcup_i (\lambda i.\,x\sqcap F\,i)$.
	\end{enumerate}
\end{lemma}

\begin{proofqed}
	By simple algebraic calculations.
\end{proofqed}

Note that every Heyting algebra embeds into its down set algebra consisting of the sets $x\down:=\lambda y.\, y\le x$.
The \emph{MacNeille completion}~\cite{macneille_partially_1937}
adding arbitrary meets and joins, while preserving existing ones, is a refinement of this embedding.

\begin{fact}[][completion_calgebra]
	\label{fact:macneille}
	Every Heyting algebra \H embeds into a complete Heyting algebra \HO, i.e. there is a function $f:\H\to \HO$ with $x\le y\leftrightarrow f\,x \le_c f\, y$ and:
	\vspace{-0.3cm}
	\begin{multicols}{2}
		\begin{enumerate}
			\item
			$f\,0\equiv 0_c$
			\item
			$f\,(x\Rightarrow y)\equiv f\,x \Rightarrow_c f\,y$
			\item
			$f\,(x\sqcap y)\equiv f\,x \sqcap_c f\,y$
			\item
			$f\,(x\sqcup y)\equiv f\,x \sqcup_c f\,y$
		\end{enumerate}	
	\end{multicols}
\end{fact}

\begin{proofqed}
	Given a set $X:\H\to \Prop$, we define the sets $\LB \,X:=\lambda x.\,\forall y\in X.\,x\le y$ of lower bounds and $\UB \,X:=\lambda x.\,\forall y\in X.\,y\le x$ of upper bounds of $X$.
	We say that a set $X$ is down-complete if $\LB\,(\UB\, X)\subseteq X$.
	Note that in particular down sets $x\down$ are down-complete and that down-complete sets are downwards closed, i.e. satisfy $x\in X$ whenever $x\le y$ for some $y\in X$.
	
	Now consider the type $\H_c:=\Sigma X.\,\LB\,(\UB\, X)\subseteq X$ of down-complete sets preordered by set inclusion $X\subseteq Y$.
	It is immediate by construction that the operation $\bigsqcap_c P:=\bigcap P$ defines arbitrary meets in \HO.
	Moreover, it is easily verified that further setting
	\begin{small}
		$$0_c:=0\down \hspace{0.3cm} X\sqcap_c Y:= X\cap Y\hspace{0.3cm}X\sqcup_c Y:= \LB\,(\UB\,(X\cup Y))\hspace{0.3cm}X\Rightarrow_c Y:=\lambda x.\,\forall y\in X.\, x\sqcap y \in Y$$
	\end{small}
	turns \HO into a (hence complete) Heyting algebra.
	The only non-trivial case is implication, where $X\Rightarrow_c Y\,\equiv\, \bigsqcap_c(\lambda Z.\,\exists x\in X.\,Z\equiv (\lambda y.\,y\sqcap x\in Y))$ is a helpful characterisation to show that $X\Rightarrow_c Y$ is down-complete whenever $Y$ is.
	
	Finally, $x\down$ clearly is a structure preserving embedding as specified.
\end{proofqed}

We now define how formulas can be evaluated in a complete Heyting algebra.

\setCoqFilename{Lindenbaum}
\begin{definition}[][hsat]
	\label{def:heval}
	Given a complete Heyting algebra \H we extend interpretations $\heval \_ : \forall P:\Preds.\,\Term^{|P|}\to \H$
	of atoms to formulas using size recursion by 
	\begin{align*}
		\heval{\dot\bot} &:=0&
		\heval{\phi\dot\land\psi}&:= \heval\phi \sqcap \heval\psi&
		\heval{\dot\forall \phi} &:=\textstyle\bigsqcap_t\heval{\phi[t]} \\
		\heval{\phi\dot\to\psi}&:= \heval\phi \Rightarrow \heval\psi&
		\heval{\phi\dot\lor\psi}&:= \heval\phi \sqcup \heval\psi&
		\heval{\dot\exists\phi} &:=\textstyle\bigsqcup_t\heval{\phi[t]}
	\end{align*}
	and to contexts by $\heval\Gamma := \bigsqcap \lambda x.\,\exists \phi\in \Gamma.\,x=\heval\phi $.
	A formula $\phi$ is \emph{valid in} \H whenever $x\le\heval \phi$ for all $x:\H$.
\end{definition}

Note that $\heval \phi$ is defined by size recursion to account for the substitution $\phi[t]$ needed in the quantifier cases.

We first show that intuitionistic ND is sound for this semantics.

\begin{fact}[][Soundness']
	$\Gamma\vdash \phi$ implies $\forall \sigma.\,\heval{\Gamma[\sigma]}\le \heval{\phi[\sigma]}$ in every complete Heyting algebra.
\end{fact}

\begin{proof}
	By induction on $\Gamma\vdash \phi$, all cases but (DE) and (EE) are trivial.
	\begin{itemize}
		\item
		(DE) In this case $\sigma$ is not instantiated, so we leave out the annotations $[\sigma]$ for better readability.
		Suppose that $\heval\Gamma\le \heval\phi \sqcup \heval\psi$, $\heval{\Gamma,\phi}\le \heval\theta$, and $\heval{\Gamma,\psi}\le \heval\theta$, we show that $\heval\Gamma\le\heval\theta$.
		Applying the first consequence mentioned in \Cref{lem:distr}, it suffices to show $(\heval\Gamma\sqcap\heval\phi) \sqcup (\heval\Gamma\sqcap\heval\psi)\le\heval\theta$. This means to show both $\heval\Gamma\sqcap\heval\phi\le\heval\theta$ and $\heval\Gamma\sqcap\heval\psi\le\heval\theta$ which both follow from the assumptions.
		\item
		(EE) Suppose that $\forall \sigma.\,\heval {\Gamma[\sigma]} \le \bigsqcup_t\heval {\phi[t;\sigma]}$ and $\forall \sigma.\,\heval{\up \Gamma[\sigma],\phi[\sigma]}\le\heval{\up \psi[\sigma]}$, we show that $\heval {\Gamma[\sigma]}\le \heval{\psi[\sigma]}$ for a fixed $\sigma$.
		Now applying the second consequence mentioned in \Cref{lem:distr}, it suffices to show $\bigsqcup_t(\heval{\Gamma[\sigma]}\sqcap \heval {\phi[t;\sigma]})\le \heval{\psi[\sigma]}$.
		This means to show $\heval{\Gamma[\sigma], \phi[t;\sigma]}\le \heval{\psi[\sigma]}$ for all terms $t$, which follows from the second assumption instantiated with $t;\sigma$ and the observation that $\up \Gamma[t;\sigma]=\Gamma$ and $\up \phi[t;\sigma]=\phi$.\qed
	\end{itemize}
\end{proof}

\begin{corollary}[][Soundness]
	$\Gamma\vdash \phi$ implies $\heval{\Gamma}\le \heval{\phi}$ in every complete Heyting algebra.
\end{corollary}

Secondly turning to completeness, a strategy reminiscent to the case of Kripke semantics can be employed by exhibiting a universal structure, the so-called \emph{Lindenbaum algebra}, that exactly coincides with provability.

\begin{fact}[][lb_alg]
	\label{fact:lindenbaum}
	The type \Form of formulas together with the preorder $\phi\vdash \psi$ and the logical connectives as corresponding algebraic operations form a Heyting algebra.
\end{fact}

\begin{proofqed}
	Straightforward using weakening.
\end{proofqed}

We write \MCL{L} for the Lindenbaum algebra (\Cref{fact:lindenbaum}) and $\overline{\MCL{L}}$ for its MacNeille completion (\Cref{fact:macneille}).
Formulas are evaluated in $\overline{\MCL{L}}$ according to \Cref{def:heval} using the syntactic atom interpretation $\heval{P\,\vec t}:=(P\,\vec t\,)\down$.
Since $\overline{\MCL{L}}$ preserves the meets and joins of \MCL{L}, evaluation in $\overline{\MCL{L}}$ yields the set of sufficient preconditions.

\begin{lemma}[][lindenbaum_hsat]
	\label{lem:down}
	Evaluating $\phi$ in $\overline{\MCL{L}}$ yields the set of all $\psi$ with $\psi\vdash \phi$, i.e. $ \heval{\phi}\equiv \phi\down$.
\end{lemma}

\begin{proof}
	By size induction on $\phi$.
	The case for atoms is by construction and the cases for all connectives but the quantifiers are immediate since $\Downarrow$ preserves the structure of \MCL{L} as specified in \Cref{fact:macneille}.
	The quantifiers are handled as follows:
	\begin{itemize}
		\item
		($\forall$) Let $\psi\in \bigsqcap_t \heval{\phi[t]}$, we show $\up\psi\vdash\phi$ in order to establish $\psi\vdash \dot\forall \phi$.
		By \Cref{lem:named_equiv} we know that there is a fresh variable $x$ such that $\up\psi\vdash\phi$ if $\psi\vdash \phi[x]$.
		The latter follows by induction for $\phi[x]$ since $\psi\in \heval{\phi[x]}$ by assumption.
		
		Conversely, let $\psi\vdash \dot\forall \phi$, we show $\psi \in\heval{\phi[t]}$ for every term $t$ in order to establish $\psi\in \bigsqcap_t \heval{\phi[t]}$.
		By (AE) we have $\psi\vdash \phi[t]$ and conclude $\psi \in\heval{\phi[t]}$ using the inductive hypothesis for $\phi[t]$.
		
		\item
		($\exists$) Let $\psi\in \bigsqcup_t \heval{\phi[t]}$, we want $\psi\in(\exists\phi)\down$.
		Hence it suffices to show $\bigsqcup_t \heval{\phi[t]}\subseteq (\exists\phi)\down$ which reduces to $\heval{\phi[t]}\subseteq(\exists\phi)\down$ for every $t$.
		By induction we know that $\heval{\phi[t]}\equiv \phi[t]\down$ and conclude $\phi[t]\down\subseteq(\exists\phi)\down$ since $\phi[t]\vdash \exists \phi$.
		
		Conversely, let $\psi\vdash\exists \phi$, we show that $\psi\in \bigsqcup_t \heval{\phi[t]}$.
		By construction of $\bigsqcup$ we have to show that $\psi\in X$ for all down-closed $X$ with $\forall t.\, \heval{\phi[t]}\subseteq X$.
		By down-closedness it suffices to show $\psi\in\LB\,(\UB\,X)$ and hence $\psi\vdash \theta$ for $\theta \in \UB\,X$.
		Applying (EE), this reduces to $\up\psi,\phi\vdash \up\theta$ and, employing \Cref{lem:named_equiv}, to $\psi,\phi[x]\vdash\theta$ for some fresh $x$.
		This follows since already $\phi[x]\vdash\theta$ given that $\phi[x]\in \phi[x]\down \equiv \heval{\phi[x]}\subseteq X$ and $\theta \in \UB\,X$.
		\qed
	\end{itemize}
\end{proof}

\begin{theorem}[][hcompleteness]
	If ${\phi}$ is valid in every complete Heyting algebra, then $\vdash \phi$.
\end{theorem}

\begin{proofqed}
	If $\phi$ is valid, then \Cref{lem:down} implies that $\psi\vdash \phi$ forall $\psi$.
	By e.g. choosing the tautology $\psi:=\dot\bot\dot\to\dot\bot$ we can derive $\vdash \phi$ since obviously $\vdash \dot\bot\dot\to\dot\bot$.
\end{proofqed}

Switching to classical logic, we call a Heyting algebra \emph{Boolean} if it satisfies $(x\!\Rightarrow\! y)\!\Rightarrow\!x\le x$ for all $x$ and $y$, hence directly accommodating Peirce's law~(P).
Then first, classical deduction is sound for interpretation in Boolean algebras.

\begin{fact}[][BSoundness]
	$\Gamma\vdash_c \phi$ implies $\heval{\Gamma}\le \heval{\phi}$ in every complete Boolean algebra.
\end{fact}

\begin{proofqed}
	As in \Cref{coq:Soundness}, the classical rule (P) is sound by definition.
\end{proofqed}

\newcommand{\prv}{\vdash_i}
\newcommand{\prva}[2]{#1\prv #2}

Secondly, we establish the completeness of classical deduction by generalising the previous proof to all deduction systems subsuming intuitionistic ND.
So we fix a predicate $\prv~: \List(\Form)\to \Form\to \Prop$ satisfying the rules of intuitionistic ND (\Cref{def:ND}), weakening (\Cref{lem:weak}), as well as the equivalences concerning fresh variables stated in \Cref{lem:named_equiv}, and replay the construction from before.

\begin{fact}[][glb_alg]
	\label{fact:glindenbaum}
	The type \Form of formulas together with the preorder $\prva{\phi}{\psi}$ and the logical connectives as corresponding algebraic operations form a Heyting algebra.
\end{fact}

We denote the Lindenbaum algebra of $\prv$ by $\MCL{L}_i$ and its completion by $\overline{\MCL{L}_i}$.

\begin{lemma}[][glindenbaum_hsat]
	\label{lem:gdown}
	Evaluating $\phi$ in $\overline{\MCL{L}_i}$ yields the set of all $\psi$ with $\prva{\phi}{\psi}$.
\end{lemma}

If we instantiate $\prv$ with $\vdash_c$ we can conclude completeness as follows:

\begin{lemma}[][boolean_completion]
	The MacNeille completion of a Boolean algebra is Boolean.
\end{lemma}

\begin{theorem}[][bcompleteness]
	If ${\phi}$ is valid in every complete Boolean algebra, then $\vdash_c \phi$.
\end{theorem}

\begin{proofqed}
	By \Cref{coq:boolean_completion}, $\overline{\MCL{L}_c}$ is Boolean since $\MCL{L}_c$ is so due to the classical rule (P).
	Then from $\phi$ valid in $\overline{\MCL{L}_c}$ we can deduce $\vdash_c \phi$ with \Cref{lem:gdown} as before.
\end{proofqed}

Note that this general construction could of course be instantiated to intuitionistic ND in order to derive \Cref{coq:hcompleteness} in the first place, same as to other intermediate logics that are not considered in this paper.

\section{Dialogue Game Semantics}
\label{sec:dialogues}
\setCoqFilename{Sorensen}

Dialogues are games modeling a proponent defending the validity of a formula
against an opponent. In the terminology of Felscher~\cite{FelscherDialogues},
the dialogues we consider in this section are the intuitionistic E-dialogues,
generalised over their local rules $(\Form, \Form^a, \attacks, \rhd,
\defenses{-})$. Given abstract types for formulas $\Form$ and attacks $\attacks$, the
relation $\attack{\varphi}{\psi}{a}$ states that a player may attack $\varphi :
\Form$ with $a : \attacks$ by possibly admitting a unique $\psi : \Opt(\Form)$. If $\psi
= \emptyset$, no admission is made. Each $a : \attacks$ has an associated set
$\defenses{a}$ of formulas that may be admitted to fend off $a$. Special rules
restrict when the proponent may admit atomic formulas, members of the set
$\Form^a$. We write $\attackN{\varphi}{a}$ for $\attack{\varphi}{\emptyset}{a}$.
The local rules of first-order logic are given below with atomic formulas
$\Form^a := \{P\,\vec{t} ~|~ P : \Preds\}$.
\begin{small}
	\begin{align*}
	\attackN{\varphi \dot\vee \psi}{a_{\dot\vee}} \quad & \defenses{a_{\dot\vee}} = \{\varphi, \psi\} &
	\hspace{-1.5em}\attack{\varphi \dot\to \psi}{\some{\varphi}}{a_{\dot\to}} \quad & \defenses{a_{\dot\to}} = \{\psi\} &
	\attackN{\varphi \dot\wedge \psi}{a_L} \quad & \defenses{a_L} = \{\varphi\} \\
	\attackN{\dot\forall \varphi}{a_t} \quad & \defenses{a_t} = \{\varphi[t]\} &
	\attackN{\dot\bot}{a_{\dot\bot}} \quad & \defenses{a_{\dot\bot}} = \{\} &
	\attackN{\varphi \dot\wedge \psi}{a_R} \quad & \defenses{a_R} = \{\psi\} \\
	\attackN{\dot\exists \varphi}{a_{\dot\exists}} \quad & \defenses{a_{\dot\exists}} = \{ \varphi[t] ~|~ t : \Term \}
	\end{align*}
\end{small}
In contrast to their usual presentation as sequences of alternating moves, we
define dialogues as state transition systems over members $(A_o, c)$ of the type $\List(\Form) \times
\attacks$ containing the opponent's admissions ($A_o$) and last attack ($c$). The proponent
opens each round by picking a move. She can defend against the opponent's attack
$c$ by admitting a justified defense formula $\varphi \in \defenses{c}$, meaning
$\varphi \in \Form^a$ implies $\varphi \in A_o$. Alternatively, she can launch an
attack $a$ against any of the opponent's admissions if the admission resulting
from $a$ is justified.
\begin{gather*}
  \hspace{2em}
  \infer[\small\textnormal{PD}]{\epmove{(A_o, c)}{\varphi}}{\varphi \in \defenses{c} \quad \just{A_o}{\varphi}}
  \hspace{3em}
  \infer[\small\textnormal{PA}]{\epmove{(A_o, c)}{(a, \varphi)}}{\varphi \in A_o \quad \attack{\varphi}{\psi}{a} \quad \just{A_o}{\psi}}
\end{gather*}

Given such a move $m$, the opponent reacts to it by transforming the
state~$s$ into $s'$ (written as $\eomove{s\,;\,m}{s'}$). The opponent may
attack the proponent's defense formula (OA), defend against her attack (OD) or
counter her attack by attacking her admission (OC). We define $\some{\varphi} ::
A := \varphi :: A$ and $\emptyset :: A := A$.
\begin{gather*}
  \hspace{3em}\infer[\small\textnormal{OA}]{\eomove{(A_o, c)\,;\varphi}{(\psi :: A_o, c')}}{\attack{\varphi}{\psi}{c'}} \hspace{3em}
  \infer[\small\textnormal{OD}]{\eomove{(A_o, c)\,;(a, \varphi)}{(\psi :: A_o, c)}}{\psi \in \defenses{a}} \\[0.2cm]
  \hspace{3em}\infer[\small\textnormal{OC}]{\eomove{(A_o, c)\,;(a, \varphi)}{(\theta :: A_o, c')}}{\attack{\varphi}{\some{\psi}}{a} \quad \attack{\psi}{\theta}{c'}}
\end{gather*}

 A formula
$\varphi$ is then considered E-valid if it is non-atomic and for all $\attack{\varphi}{\psi}{c}$,
there is a winning strategy $\eWin{([\psi], c)}$ as defined below. %
$$ \infer{\eWin{s}}{\epmove{s}{m} \quad \forall s'.~\eomove{s\,;m}{s'} \to \eWin{s'}}$$

Following the strategy of \cite{SorensenDialogues}, we first prove the
soundness and completeness of the sequent calculus LJD which is defined in terms
of the same notions as the dialogues. Indeed, as witnessed in the proofs of
soundness and completeness, derivations of LJD are isomorphic to winning
strategies, the R- and L-rule corresponding to a proponent defense and attack,
their premises matching the possible opponent responses to each move. The
statement $\Gamma \LJD \mathcal{S}$ means that the context $\Gamma$ entails the
disjunction of the formulas contained in the set $\mathcal{S}$.
	\begin{mathpar}
	\inferrule*[Right=R]{\varphi \in \mathcal{S} \and \just{\Gamma}{\varphi} \\\\ \forall a' \theta.~
		\attack{\varphi}{\theta}{a'} \,\to\,\Gamma, \theta \LJD \defenses{a}}{\Gamma \LJD \mathcal{S}}
  
	\inferrule*[Right=L]{\varphi \in \Gamma \and \just{\Gamma}{\psi} \\\\
    \attack{\varphi}{\psi}{a} \and \forall
		\,\theta \in \defenses{a}.~\Gamma, \theta \LJD \mathcal{S} \and \forall a' \theta.~
		\attack{\psi}{\theta}{a'}\,\to\, \Gamma, \tau \LJD \defenses{a'}}{\Gamma \LJD \mathcal{S}}
	\end{mathpar}

\begin{theorem}[][eequiv]\label{thm:esoundness}
 Any formula $\varphi$ is E-valid if and only if one can derive $[] \LJD \{\varphi\}$.
\end{theorem}
\begin{proofqed}
  $\eWin{(A_o, c)} \to A_o \LJD \defenses{c}$ holds by induction on $\eWin{(A_o, c)}$.
  From this, completeness follows with an application of the $R$-rule,
  transforming a winning strategy $\eWin{([\psi], c)}$ for any $\attack{\varphi}{\psi}{c}$ into $[\psi] \LJD
  \defenses{c}$.
  Soundness can be proven symmetrically.
\end{proofqed}

\setCoqFilename{DialogFull}
To arrive at a more traditional soundness and completeness result, we show that
one can translate between derivations in LJD and the intuitionistic sequent
calculus LJ deriving sequents $\Gamma \fsq \varphi$ as defined in \Cref{def:lj} of \Cref{sec:systems}.

\begin{lemma}[][Dprv_fprv_equiv]
  One can derive $\Gamma \LJD \{\varphi\}$ if and only if one can derive $\Gamma \fsq \varphi$.
\end{lemma}
\begin{proofqed}
  Completeness is generalised as below and shown by induction on $\Gamma\hspace{-0.3em} \LJD\hspace{-0.3em} \mathcal{S}$:
  $$\Gamma \LJD \mathcal{S} ~\to~ \forall \varphi.~(\forall \psi, \Gamma \subseteq
  \Gamma'.~\Gamma' \fsq \psi \to \Gamma' \fsq \varphi) \to \Gamma \fsq \varphi$$
  Soundness follows analogously from $\Gamma \fsq \varphi \to \forall \sigma.~\Gamma[\sigma]
  \LJD \{\varphi[\sigma]\}$.
\end{proofqed}
\begin{corollary}[][evalid_fprv_equiv]
  Any formula $\varphi$ is E-valid if and only if one can derive $[] \fsq \varphi$.
\end{corollary}

\setCoqFilename{Sorensen}
We now extend the soundness and completeness results to D-dialogues, which lift
the E-dialogues' restriction on the opponent to only react to the directly
preceding proponent move. We formalise D-dialogues as a state transition system
over $(A_p, C_p, A_o, C_o) : \dstate$ where $A_{-}$ contains the open admissions
and $C_{-}$ the unanswered challenges against the respective players. As before,
the proponent may defend against the last open challenge against her (PD) or
attack one of the opponent's admissions (PA).
\begin{gather*}
  \hspace{3em}\infer[\small\textnormal{PD}]{\dpmove{(A_p, c :: C_p, A_o, C_o)}{(\varphi :: A_p,
      C_p, A_o, C_o)}}{\varphi \in \defenses{c} \quad \just{A_o}{\varphi}} \\[0.2cm]
  \hspace{3em}\infer[\small\textnormal{PA}]{\dpmove{(A_p, C_p, A_o, C_o)}{(\psi :: A_p, C_p, A_o, a
      :: C_o)}}{\varphi \in A_o \quad \just{A_o}{\psi} \quad \attack{\varphi}{\psi}{a}}
\end{gather*}
Symmetrically, the opponent may do the same, although she may still only ever
attack each of proponent's admissions once.
\begin{gather*}
  \hspace{3em}\infer[\small\textnormal{OD}]{\domove{(A_p, C_p, A_o, a :: C_o)}{(A_p,
      C_p, \varphi :: A_o, C_o)}}{\varphi \in \defenses{a}} \\[0.2cm]
  \hspace{3em}\infer[\small\textnormal{OA}]{\domove{(A_p \app \varphi :: A_p',
      C_p, A_o, C_o)}{(A_p \app A_p', c :: C_p, \psi :: A_o, C_o)}}{\attack{\varphi}{\psi}{c}}
\end{gather*}

Winning strategies and validity for D-dialogues are defined completely
analogously to those for E-dialogues. As a winning strategy for D-dialogues
contains information on how to fend off strictly more attacks than its
E-counterpart, the proof strategy of completeness of LJD with regards to E-strategies
can be extended to D-strategies.
\begin{lemma}[][dcompleteness]
  If $\varphi$ is D-valid, one can derive $[] \LJD \{\varphi\}$.
\end{lemma}
\begin{proofqed}
  We show $\dWin{( A_p, c :: C_p, A_o, C_o )} \to A_o \LJD \defenses{c}$ by
  induction on $\dWin{( A_p, c :: C_p, A_o, C_o )}$. The result follows from the
  definitions of validity.
\end{proofqed}

Proving LJD sound for D-dialogues is more involved. First, we make an
observation about the structure of a winning strategy for D-dialogues that was
derived from an LJD derivation: If the derivation ends in an application of (L), telling the
proponent to attack one of the opponent's admissions via an attack $a$, it only
indicates how to continue to fend off the current challenge against the
proponent after the opponent admits some formula from $\defenses{a}$.
By attacking one of the opponent's admissions, the proponent thus defers the
continuation of her ``current line of argument'' in the resulting D-strategy until the opponent chooses to
react, which by the rules of D-dialogues may be later than next turn.

We define a new class of dialogues, the S-dialogues, which encode this
observation and aiding us in stating appropriate invariants in the proof of
soundness.
We formalise them as a transition system over $(A_p, A_o, D) : \sstate$
where a pair $(a, c) \in D$ represents the proponent deferring her response to $c$
until the opponent responds to $a$. The proponent, along a state
$(A_p, A_o, D)$, is given a current challenge $c$ to react to. She can
either defend against that challenge (PD) or defer responding by attacking one of
the opponent's admissions (PA).
\begin{gather*}
  \hspace{3em}\infer[\small\textnormal{PD}]{\spmove{(A_p , A_o, D) \,;\, c}{(\varphi :: A_p, A_o,
      D)}}{\varphi \in \defenses{c} \quad \just{A_o}{\varphi}} \\[0.2cm]
  \hspace{3em}\infer[\small\textnormal{PA}]{\spmove{(A_p, A_o, D) \,;\, c}{(\psi :: A_p, A_o, (a, c)
      :: D)}}{\varphi \in A_o \quad \just{A_o}{\psi} \quad
    \attack{\varphi}{\psi}{a}}
\end{gather*}
The opponent then can either reissue the current challenge by defending against
the proponent's attack (OD) or issue a new challenge by attacking one of the
proponent's admissions (OA).
\begin{gather*}
  \infer[\small\textnormal{OD}]{\somove{(A_p, A_o, (a, c) :: D)}{(A_p,
      \varphi :: A_o, D)\,;\,c}}{\varphi \in \defenses{a}} \\[0.2cm]
  \infer[\small\textnormal{OA}]{\somove{(A_p \app \varphi :: A_p',
      A_o, D)}{(A_p \app A_p', \psi :: A_o, D)\,;\,c}}{\attack{\varphi}{\psi}{c}}
\end{gather*}

The winning strategies and notion of validity for S-dialogues are again defined
analogously to the other two kinds of dialogues. As S-dialogues are essentially
just D-dialogues with a stack structure imposed on $C_p$ and $C_o$ (hence \emph{S}-dialogues), we can
translate their winning strategies back into ones for D-dialogues. Importantly,
this means it suffices to show soundness for S-dialogues.

\begin{lemma}[][svalid_dvalid]
  Any S-valid formula $\varphi$ is also D-valid.
\end{lemma}
\begin{proofqed}
  We show $\sWin{(A_p, A_o, D) \,;\, c} \to \dWin{(A_p, c :: \pi_1\,D, A_o,
    \pi_2\,D)}$ by induction on $\sWin{}$. The claim follows from the
  definitions of validity.
\end{proofqed}

Thus what remains is to prove that LJD is sound with regards to S-validity. The
difficulty of this proof stems from the fact that the resulting winning strategy
will be much ``deeper'' than the LJD derivation because the opponent may now
perform all possible responses to a proponent move within the same game instead
of being restricted to picking only one. This in turn means that a simple
induction on the derivation will not suffice for the proof, instead requiring a
more involved induction principle. We thus first give the proof, leaving the
induction principle abstract and then define it afterwards.

\setCoqFilename{SDialogues}
\begin{theorem}[][Dprv_svalid]\label{lem:ssound}
  If $[] \LJD \{\varphi\}$ can be derived then $\varphi$ is S-valid.
\end{theorem}
\begin{proofqed}
  For this, we prove a generalised claim, namely that for all $A_p, A_o, D, c$ if
  \begin{enumerate}[(1)]
  \item for each $\varphi \in A_p$ there is a $\Gamma \subseteq A_o$ and a family
    $\forall \attack{\varphi}{a}{\psi}.~\psi :: \Gamma \LJD \defenses{a}$
  \item for each $(a, c) \in D$ there is a $\Gamma \subseteq A_o$ and a family $\forall
    \theta \in \defenses{a}.~\theta :: \Gamma \LJD \defenses{c}$
  \item there is a $\Gamma \subseteq A_o$ and a derivation $\Gamma \LJD \defenses{c}$
  \end{enumerate}
  then we can derive $\sWin{(A_p, A_o, D) \,;\,c}$. We first show by case
  distinction on the derivation of $\Gamma \LJD \defenses{c}$ that there is a
  transition $\spmove{(A_p, A_o, D) \,;\,c}{(A'_p, A_o, D')}$ such that $A'_p$
  and $D'$ satisfy invariants (1) and (2).
  \begin{enumerate}
  \item[(R)] Then there is some justified $\varphi \in \defenses{c}$ such that
    $\forall \attack{\varphi}{\psi}{a}.~\psi :: \Gamma \LJD \defenses{a}$.
    Then the proponent will defend by admitting $\varphi$.
    Invariant (1) then extends to $A'_p := \varphi :: A_p$. 
  \item[(L)] Then there is some $\varphi \in \Gamma$ and a justified attack
    $\attack{\varphi}{\psi}{a}$ such that both $\forall \theta \in
    \defenses{a}.~\theta :: \Gamma \LJD \defenses{c}$ and $\forall
    \attack{\psi}{\theta}{a'}.~\theta :: \Gamma \LJD \defenses{a'}$.
    Then the proponent will attack with $a$.
    The invariants then extend to $\psi :: A_p$ and $(a, c) :: D$.
  \end{enumerate}
  Now we show that any opponent moves from $(A_p', A_o, D')$ lead to a winning
  position, again by case distinction.
  \begin{enumerate}
  \item[(OD)] Then $D' = (a, c) :: D''$ and the opponent just admitted some
    $\theta \in \defenses{a}$. By invariant (2), there is a $\Gamma \subseteq
    A_o$ with $\theta :: \Gamma \LJD \defenses{c'}$. Then we can obtain
    $\sWin{(A'_p, \theta :: A_o, D'')\,;\, c'}$ per inductive
    hypothesis as $\theta :: \Gamma \subseteq \theta :: A_o$ and
    invariant (3) thus holds for $c'$.
  \item[(OA)] Then $A'_p = A \app \psi :: A''$ and the opponent
    just attacked with $\attack{\psi}{\theta}{c'}$. By invariant (1), there is $\Gamma \subseteq A_o$ with $\theta :: \Gamma
    \LJD \defenses{c'}$. Then we can obtain $\sWin{(A' \app A'', \theta ::
      A_o, D')\,;\,c'}$ as $\theta :: \Gamma \subseteq \theta :: A_o$ and
    invariant (3) thus holds for $c'$.
  \end{enumerate}
  Note that when applying the inductive hypothesis we are implicitly using the fact that the invariants all
  extend to $\theta :: A_o$.
\end{proofqed}
\begin{corollary}
  Any formula $\varphi$ for which $[] \LJD \{\varphi\}$ can be derived is D-valid.
\end{corollary}
\begin{corollary}
  A formula $\varphi$ is E-valid if and only if it is D-valid.
\end{corollary}

Now all that remains is finding the correct induction principle. Intuitively,
the proof above is well-founded as each ``inductive step'' replaces one of the
(families of) derivations from the invariants with its specialisation or
sub-derivations twice, once for the proponent move, once for the opponent move.
We can abstract this observation into the following relation.

\setCoqFilename{WFexp}
\begin{definition}[][tlexp]
  Let $R : X \to X \to \Prop$ be a relation on some type $X$. Then we define
  $\prec_R^+\, : \List(X) \to \List(X) \to \Prop$ to be the transitive closure of
  $\prec_R$, where $A \prec_R B$ holds iff there are $B', B'', C : \List(X)$ such
  that $B = B' \app x :: B''$, for all $c \in C$ we have $R\, c\, x$, and $A$ is a
  permutation of $B' \app C \app B''$.
\end{definition}

For the proof of \Cref{lem:ssound} we then choose $X$ to be the sum of the three
kinds of (families of) derivations used in the invariants and the relation $R$
being the specialisation and sub-derivation relation on them. The induction then
proceeds on $\prec_R^+$ on a list containing all the families asserted by the invariants.
Note that the list-splitting of $\prec_R$ exactly mirrors the act of replacing
one invariant with its specialisation or sub-derivations. The permutations are a
technical accommodation that make it easier to state the invariants formally by
allowing us to list the invariants in a fixed order as we can ``move the new
invariants into place'' after each step.
Finally, the transitivity is needed as each inductive step of \Cref{lem:ssound}
takes two such list-splitting steps, one for the proponent and one for the
opponent.

\begin{lemma}[][well_founded_tlexp]
  If $R : X \to X \to \Prop$ is well-founded then so is $\prec_R^+$.
\end{lemma}
\begin{proofqed}
  We first show that $A \prec_R' B := \exists B', B'', C.~B = B' \app x :: B''
  \wedge A = B' \app C \app B'' \wedge \forall c \in C.\,R\,c\,x$ is well-founded.
  To this end, we first prove that $\prec_R'$ being well-founded on $A$ and $B$ means it is
  well-founded on $A \app B$ by well-founded induction on $A$ and $B$ along
  $\prec_R'$. Then we can show that $\prec_R'$
  is well-founded on singletons $[x]$ by $R$-induction on $x$ as for any $Rcx$
  we know that $\prec_R'$ is well-founded on $[c]$ per inductive hypothesis
  which we can extend to arbitrary lists $C$ of such $Rcx$ using the previous
  fact. The fact can then be applied again to obtain well-foundedness of
  $\prec_R'$ on arbitrary lists. 

  Now we show that if $A \prec_R' B$ and if $B$ is a permutation of $B'$ then there is
  a permutation $A'$ of $A$ with $A' \prec_R' B'$ per induction on $B$. With this,
  we can show that $\prec_R'$ being well-founded on $A$ entails that $\prec_R$ is
  well-founded on any permutation of $A$, again by induction on $A$. From this,
  well-foundedness of $\prec_R$ follows from the well-foundedness of $\prec_R'$.

  Lastly, we use the fact that transitive closure maintains
  well-foundedness.
\end{proofqed}

\section{Discussion}
\label{sec:discussion}

We have analysed the completeness of common deduction systems for first-order logic with regards to various explanations of logical validity.
Model-theoretic semantics are the most direct implementation of the idea that terms represent objects of a domain of discourse.
Particularly in a formal meta-theory such as constructive type theory, model-theoretic completeness justifies the common practice to verify consequences of a first-order axiomatisation by studying models satisfying corresponding meta-level axioms.
However, model-theoretic semantics typically do not admit constructive completeness and, if not generalised to exploding models, require Markov's Principle as soon as falsity is involved.
Contrarily, evidence for the validity of a first-order formula in algebraic semantics and game semantics can be algorithmically transformed into syntactic derivations.

The analysis of the completeness theorem for classical first-order logic benefited from the use of constructive type theory with an impredicative universe of propositions as underlying system.
Constructive type theory has fewer built-in assumptions than the systems usually used in both classical and constructive reverse mathematics, allowing for sharper equivalence results.
In classical reverse mathematics~\cite{SimpsonRM}, one uses classical logic freely (i.e.\ $\EM$ is provable), but does not assume strong function or set existence principles, nor choice axioms, and thus $\WKL$ is not provable in the weakest considered base system $\mathsf{RCA}_0$.
In contrast, constructive reverse mathematics~\cite{ishihara_reverse_2006} is based on Bishop's constructive mathematics $\mathsf{BISH}$~\cite{BishopAnalysis} as for instance formalised by predicative type theories~\cite{martin1975intuitionistic}.
$\mathsf{BISH}$ is based on intuitionistic logic only (i.e.\ neither $\EM$ nor $\MP$ are provable), but countable and dependent choice axioms are provable\rlap,\footnote{\scriptsize The universal assumption of countable choice for constructive mathematics is criticised e.g.\ by Richman~\cite{richman2000fundamental,richmanConstructiveMathematicsChoice2001}.} turning $\WKL$ into a purely logical axiom, equivalent to the lesser limited principle of omniscience $\mathsf{LLPO}$~\cite{ishihara1990omniscience} and in particular into a consequence of $\EM$.
Thus, $\mathsf{BISH}$ and predicative type theories are insensitive to the role of $\WKL$ w.r.t.\ the completeness theorem as formulated in \Cref{tarski_theorem}.
In type theory with a universe of propositions however, $\EM$ likely does not imply $\WKL$ and thus $\EM \land \WKL$ becomes a sensible proposition with interesting equivalences.

Of course, there are more semantic accounts of first-order logic than the selection studied in this paper.
For instance, there are hybrid variants such as interpreting both terms in a model and logical operations in an algebra, or dialogues with atomic formulas represented as underlying games.
More generally, there are entirely different approaches based on realisability, the Brouwer-Heyting-Kolmogorov interpretation, or proof-theoretic semantics, all coming with interesting completeness problems worth analysing in constructive type theory.
More ideas for future work are outlined after a brief summary of related work.

\subsection{Related Work}
\label{sec:related}

Our analysis of completeness in constructive type theory was motivated by previous work~\cite{ForsterCPP}, carried out in Wehr's bachelor's thesis~\cite{Wehr2019}, and is directly influenced by multiple prior works.
In their analysis of Henkin's proof, Herbelin and Ilik~\cite{HerbelinHenkin} give a constructive
model existence proof and the constructivisation of completeness via exploding
models. Herbelin and Lee~\cite{HerbelinCut} demonstrate the constructive Kripke
completeness proof for minimal models and mention how to extend the approach to
standard and exploding models.
Scott~\cite{scott_algebraic_2008} establishes completeness of free logic interpreted in a hybrid semantics comprising model-theoretic and algebraic components.
Urzyczyn and S\o rensen~\cite{SorensenDialogues}
give a proof of dialogue completeness via generalised dialogues for
classical propositional logic.

The first proof that the completeness of intuitionistic first-order logic
entails Markov's Principle was given by Kreisel~\cite{KreiselMP}, although he
attributes the proof idea to Gödel. The proof has since inspired a range of
works deriving related non-constructivity results for different kinds of
completeness~\cite{BeradiClassical,KreiselRE,LeivantRE,McCartyNonarithmetic,McCartyIntuitionistic,McCartyMetamathematics}.
Krivtsov has analysed the necessity of $\WKL$ for completeness proofs for both classical and intuitionistic first-order logic w.r.t.\ decidable models~\cite{KrivtsovClassical,KrivtsovFan}.

The completeness of first-order logic has been mechanised in many interactive
theorem provers such as
Isabelle/HOL~\cite{BlanchetteComp,RidgeTP,SchlichtkrullResolution},
NuPRL~\cite{ConstableIFol,UnderwoodComp}, Mizar~\cite{BraselmannComp},
Lean~\cite{LeanContinuum}, and Coq~\cite{HerbelinCut,DankoThesis,gilbert:hal-01204599}. Among them,
\cite{ConstableIFol} and \cite{DankoThesis} share our focus on the
constructivity of completeness. Constable and Bickford~\cite{ConstableIFol} give
a constructive proof of completeness for the BHK-realisers of full
intuitionistic first-order logic in NuPRL. Their proof is fully constructive
when realisers are restricted to be normal terms, requiring Brouwer's fan
theorem when lifting that restriction. In his PhD thesis~\cite{DankoThesis},
Ilik mechanises multiple constructive proofs of first-order completeness in Coq.
Especially noteworthy are the highly non-standard, constructivised Kripke models for
full classical and intuitionistic first-order logic he presents in Chapters 2 and
3.
Gilbert and Hermant~\cite{gilbert:hal-01204599} describe a normalisation-by-evaluation completeness proof using Heyting algebras and implement it for propositional logic in Coq.

\subsection{Future Work}

We plan to further extend our constructive analysis and Coq library of completeness theorems to all logical connectives and to uncountable signatures, both relying on additional logical assumptions.

Concerning model-theoretic semantics, our analysis left open at least three interesting questions:
First, we have shown that completeness for the classical $\forall,\to,\bot$-fragment w.r.t.\ omniscient models and arbitrary theories is equivalent to both $\EM$ and $\WKL$.
Completeness for the $\forall,\to,\bot$-fragment w.r.t.\ omniscient models and \emph{enumerable} theories certainly implies $\MP$ and $\WKL_{\mathcal{D}}$, but it is unclear how to obtain an equivalence.
Restricting completeness to contexts implies $\MPL$, but a formulation of $\WKL_{\mathcal{D}}$ for $\L$-computable functions is equivalent to falsity, due to Kleene's tree~\cite{kleene1953recursive}.
Secondly, we only prove that completeness for classical first-order logic with all connectives is equivalent to $\EM$, but leave open what the necessary and sufficient principles are to obtain completeness for classical first-order logic with all connectives w.r.t.\ contexts or enumerable theories.
Thirdly, we have not considered decidable Kripke models nor intuitionistic completeness for full first-order logic.
Veldman's constructivisation of completeness for intuitionistic first-order logic relies on decidable models and the fan theorem (which is a consequence of $\WKL$~\cite{ishihara2006weak}) to treat disjunction~\cite{VeldmanExplosion}.
It is an interesting direction for future research whether the fan theorem can be avoided in the presence of disjunction when using propositional models as we do in this paper.

Subsequently, it would be interesting to study other aspects of model theory in the setting of constructive type theory, for instance the Löwenheim-Skolem theorems or first-order axiomatisations of arithmetic and set theory.
Another idea is to analyse the completeness of second-order logic interpreted in Henkin semantics, as this formalism suffices to express the higher-order axiomatisation of set theory studied in~\cite{Kirst2018}.
Furthermore, the contemporary syntactic presentation of dialogues we studied
differs from that first put forward by
Lorenzen~\cite{LorenzenDialogues,LorenzenDialogues2} which was distinctly more
model-theoretic, raising the question whether the constructivity of its
completeness results mirrors those for other model theoretic semantics within this work.
Lastly, we conjecture that $\MPL$ is strictly weaker than $\MP$, but are not aware of a proof.

\subsubsection*{Acknowledgments}
We thank Kathrin Stark for adapting Autosubst according to our needs,
Fabian Kunze for helping with technicalities during the mechanisation of \Cref{coq:iprv_halt}, and Hugo Herbelin for fruitful discussion and pointers to relevant work.
We also thank the anonymous reviewers whose comments helped improving the final version of this paper.

\appendix

\section{Notes on the Coq Mechanisation}
\label{sec:Coq}

Our mechanisation consists of about 9k lines of code, with an even split between specification and proofs.
The code is structured as follows.

\begin{center}
	\begin{tabular}{|l|r|r|}
		\hline
		\textbf{Section} & \textbf{Specification} & \textbf{Proofs} \\
		\hline
		Preliminaries Autosubst & 169 & 53 \\
		\hline
		Preliminaries for $\Form^*$ & 680 & 599 \\
		\hline
		Tarski Semantics & 655 & 682 \\
    \hline
    Extended Tarski Semantics & 130 & 203 \\
    \hline
    Compactness and $\WKL$ & 266 & 588 \\
		\hline
		Kripke Semantics & 342 & 255 \\
		\hline
		On Markov's Principle & 593 & 978 \\
		\hline
		Preliminaries for $\Form$ & 523 & 430 \\
		\hline
		Algebraic Semantics & 349 & 570 \\
		\hline
		Dialogue Semantics & 563 & 539 \\
		\hline
		\textbf{Total} & 4270 & 4897 \\
		\hline
	\end{tabular}
\end{center}

In general, we find that Coq provides the ideal grounds for mechanising projects like ours.
It has external libraries supporting the mechanisation of syntax, enough automation to support the limited amounts we need and allows constructive reverse mathematics due to its axiomatic minimality.

In the remainder of the section, we elaborate on noteworthy design choices of the mechanisation.

\paragraph{Formalisation of binders}

There are various competing techniques to mechanise binders in proof assistants.
In first-order logic, binders occur in quantification.
The chosen technique especially affects the definition of deduction systems and can considerably ease or impede proofs of standard properties like weakening.

We opted for a de Bruijn representation of variables and binders with parallel substitutions.
The Autosubst 2 tool~\cite{AutoSubst2} provides convenient automation for the definition of and proofs about this representation of syntax.

Notably, our representation then results in very straightforward proofs for weakening with only 5 lines.
In contrast, using other representations for binders results in considerably more complicated weakening proofs, e.g. 150 lines in an approach using names~\cite{ForsterCPP} and 95 lines in an approach using traced syntax~\cite{HerbelinCut}.

Also note that first-order logic has the simplest structure of binders possible:
Since quantifiers range over terms, but terms do not contain binders, we do not need a prior notion of renaming, as usually standard in de Bruijn presentations of syntax.
This observation results in more compact code (because usually, every statement on substitutions has to be proved for renamings first, with oftentimes the same proof) and was incorporated into Autosubst 2, which now does not generate renamings if they are not needed.
Furthermore, we remark that the HOAS encoding of such simple binding structures results in a strictly positive inductive type and would thus be in principle definable in Coq.

\paragraph{Formalisation of signatures}

Our whole development is parametrised against a signature, defined as a typeclass in Coq:%
\begin{verbatim}
Class Signature := B_S { Funcs : Type; fun_ar  : Funcs -> nat ;
Preds : Type; pred_ar : Preds -> nat }.
\end{verbatim}
We implement term and predicate application using the dependent vector type.
While the vector type is known to cause issues in dependent programming, in this instance it was the best choice.
Recursion on terms is accepted by Coq's guardness checker, and while the generated induction principle (as is always the case for nested inductives) is too weak, a sufficient version can easily be implemented by hand:

\begin{verbatim}
Inductive vec_in (A : Type) (a : A) : forall n, vector A n -> Type :=
| vec_inB n (v : vector A n) : vec_in a (cons a v)
| vec_inS a' n (v :vector A n) : vec_in a v -> vec_in a (cons a' v).

Lemma strong_term_ind (p : term -> Type) :
(forall x, p (var_term x)) -> 
(forall F v, (forall t, vec_in t v -> p t) -> p (Func F v)) -> 
forall (t : term), p t.
\end{verbatim}

\paragraph{Syntactic fragments}

There are essentially four ways to mechanise the syntactic fragment $\Form^*$.
First, we could parametrise the type of formulas with tags, as done in~\cite{ForsterCPP}, or abstract types of connectives, as done in~\cite{KirstLarchey-Wendling:2020:Trakhtenbrot}, and second, we could use well-explored techniques for modular syntax in Coq~\cite{keuchel2013generic,delaware2013meta,forster2020coqala}.
However, both of these approaches would not be compatible with the Autosubst tool.
Additionally, modular syntax would force users of our developed library for first-order logic to work on the peculiar representation of syntax using containers or functors instead of regular inductive types.

The third option is to only define the type $\Form$, and then define a predicate on this formulas characterising the fragment $\Form^*$.
This approach introduces many additional assumptions in almost all statements, decreasing their readability and yielding many simple but repetitive proof obligations.
Furthermore, we would have to parameterise natural deduction over predicates as well, in order for the (IE) rule to not introduce terms e.g.~containing $\dot\exists$ when only deductions over $\Form^*$ should be considered.

To make the mechanisation as clear and reusable as possible, we chose the fourth and most simple possible approach:
We essentially duplicate the contents of \Cref{sec:prelims} for both $\Form^*$ and $\Form$, resulting in two independent developments on top of the two preliminary parts.

\paragraph{Parametrised deduction systems}

When defining the minimal, intuitionistic, and classical versions of natural deduction, a similar issue arises. %
Here, we chose to use one single predicate definition, where the rules for explosion and Peirce can be enabled or disabled using tags, which are parameters of the predicate.
\begin{verbatim}
Inductive peirce := class | intu. 
Inductive bottom := expl  | lconst.
Inductive prv : forall (p : peirce) (b : bottom), 
list (form) -> form -> Prop := (* ... *).
\end{verbatim}
We can then define all considered variants of ND by fixing those parameters:
\begin{verbatim}
Notation "A ⊢CE phi" := (@prv class expl A phi) (at level 30).
Notation "A ⊢CL phi" := (@prv class lconst A phi) (at level 30).
Notation "A ⊢IE phi" := (@prv intu expl A phi) (at level 30).
\end{verbatim}

This definition allows us to give for instance a general weakening proof, which can then be instantiated to the different versions.
Similarly, we can give a parametrised soundness proof, and depending on the parameters fix required properties on the models used in the definition of validity.

\paragraph{Object tactics}

At several parts of our developments we have to build concrete ND derivations.
This can always be done by explicitly applying the constructors of the ND predicate, which however becomes tedious quickly.
We thus developed object tactics reminiscent of the tactics available in Coq.
The tactic \texttt{ointros} for instance applies the (II) rule, whereas the tactic \texttt{oapply} can apply hypotheses, i.e. combine the rules (IE) and (C).
All object tactics are in the file \texttt{FullND.v}.

\paragraph{Extraction to $\lambda$-calculus}

The proof that completeness of provability w.r.t.~standard Tarski and Kripke semantics is equivalent to $\MPL$ crucially relies on an $\L$-enumeration of provable formulas.
While giving a Coq enumeration is easy using techniques described in~\cite{ForsterCPP}, the translation of any function to a model of computation is considered notoriously hard.
We use the framework by Forster and Kunze~\cite{forster_et_al:LIPIcs:2019:11072} which allows the automated translation of Coq functions to $\L$.

Using the framework was mostly easy and spared us considerable mechanisation effort.
However, the framework covers only simple types, whereas our representation of both terms and formulas contains the dependent vector type.
We circumvent this problem by defining a non-dependent term type \texttt{term'} and a predicate \texttt{wf} characterising exactly the terms in correspondence with our original type of terms.
\begin{verbatim}
Inductive term' := var_term' : nat -> term' | Func' (name : nat) 
| App' : term' -> term' -> term'.

Inductive varornot := isvar | novar.
Inductive wf : varornot -> term' -> Prop :=
| wf_var n : wf isvar (var_term' n) 
| wf_fun f : wf novar (Func' f) 
| wf_app v s t : wf v s -> wf novar t -> wf novar (App' s t).
\end{verbatim}

We then define a formula type \texttt{form'} based on \texttt{term'} and a suitable deduction system.
One can give a bijection between well-formed non-dependent terms \texttt{term'} and dependent terms \texttt{term} and prove the equivalence of the corresponding deduction systems under this bijection.

Functions working on \texttt{term'} and \texttt{form'} were easily extracted to $\L$ using the framework, yielding an $\L$-enumerability proof for ND essentially with no manual mechanisation effort.

\paragraph{Usage of Axioms}

As the aim of this project is to analyse the minimal assumptions underlying completeness theorems, our mechanisation is in principle set up such to not introduce additional axioms.
Sole exception is the axiom of functional extensionality, which is currently required by the Autosubst tool to keep the proof terms small when rewriting with point-wise equal substitutions.
Autosubst could of course be extended with a mode using setoid rewriting instead of appealing to functional extensionality, and if willing to waive the tool support, one can manually mechanise first-order logic axiom-free as done in~\cite{KirstLarchey-Wendling:2020:Trakhtenbrot}.

\paragraph{Library of mechanised undecidable problems in Coq}

We take the mechanisation of synthetic undecidability from~\cite{ForsterCPP}, which is part of the Coq library of mechanised undecidable problems~\cite{forster2020coq}.
The reduction from $\L$-halting to provability is factored via Turing machines, Minsky machines, binary stack machines and the Post correspondence problem (PCP), all part of the library as well.

\paragraph{Equations package}

Defining non-structurally recursive functions is sometimes considered hard in Coq and other proof assistants based on dependent type theory.
One such example is the function $\heval{\_}$ used to embed formulas into Heyting algebras (\Cref{def:heval}).
We use the Equations package~\cite{sozeau2019equations} to define this function by recursion on the size of the formula, ignoring terms.
The definition then becomes entirely straightforward and the provided \texttt{simp} tactic, while sometimes a bit premature, enables compact proofs. 

\section{Overview of Deduction Systems}
\label{sec:systems}

\setCoqFilename{FullND}
\begin{definition}[][prv]
	\label{def:ND}
	Intuitionistic natural deduction is defined by the following rules:
	\begin{gather*}
	\infer[\small\textnormal{C}]{\Gamma\vdash \phi}{\phi\in \Gamma}\hspace{3em}
	\infer[\small\textnormal{E}]{\Gamma\vdash \phi}{\Gamma\vdash \dot{\bot}}\hspace{3em}
	\infer[\small\textnormal{II}]{\Gamma\vdash \phi\dot{\to}\psi}{\Gamma,\phi\vdash \psi}\hspace{3em}
	\infer[\small\textnormal{IE}]{\Gamma\vdash \phi}{\Gamma\vdash \phi\dot{\to}\psi&\Gamma\vdash\phi}
	\\[0.3cm]
	\infer[\small\textnormal{CI}]{\Gamma\vdash \phi\dot{\land}\psi}{\Gamma\vdash \phi& \Gamma\vdash \psi}\hspace{3em}
	\infer[{\small\textnormal{CE}_1}]{\Gamma\vdash \phi}{\Gamma\vdash \phi\dot{\land}\psi}\hspace{3em}
	\infer[{\small\textnormal{CE}_2}]{\Gamma\vdash \psi}{\Gamma\vdash \phi\dot{\land}\psi}
	\\[0.3cm]
	\infer[{\small\textnormal{DI}_1}]{\Gamma\vdash \phi\dot{\lor}\psi}{\Gamma\vdash \phi}\hspace{3em}
	\infer[{\small\textnormal{DI}_2}]{\Gamma\vdash \phi\dot{\lor}\psi}{\Gamma\vdash \psi}\hspace{3em}
	\infer[\small\textnormal{DE}]{\Gamma\vdash\theta}{\Gamma\vdash\phi\dot{\lor}\psi&\Gamma,\phi\vdash \theta&\Gamma,\psi\vdash \theta}
	\\[0.3cm]
	\infer[\small\textnormal{AI}]{\Gamma\vdash\dot{\forall} \phi}{\up \Gamma\vdash \phi}\hspace{3em}
	\infer[\small\textnormal{AE}]{\vphantom{\dot{\forall}}\Gamma\vdash \phi[t]}{\Gamma\vdash \dot \forall \phi}\hspace{3em}
	\infer[\small\textnormal{EI}]{\Gamma\vdash \dot{\exists}\phi}{\Gamma\vdash \phi[t]}\hspace{3em}
	\infer[\small\textnormal{EE}]{\vphantom{\dot{\forall}}\Gamma\vdash \psi}{\Gamma\vdash\dot{\exists}\phi&\up \Gamma,\phi\vdash \up \psi}
	\end{gather*}
	We write $\vdash\phi$ whenever $\phi$ is intuitionistically provable from the empty context.
\end{definition}

\begin{definition}[][prv]
	\label{def:CND}
	Classical natural deduction is defined by the following rules:
	\begin{gather*}
	\infer[\small\textnormal{C}]{\Gamma\vdash_c \phi}{\phi\in \Gamma}\hspace{3em}
	\infer[\small\textnormal{E}]{\Gamma\vdash_c \phi}{\Gamma\vdash_c \dot{\bot}}\hspace{3em}
	\infer[\small\textnormal{II}]{\Gamma\vdash_c \phi\dot{\to}\psi}{\Gamma,\phi\vdash_c \psi}\hspace{3em}
	\infer[\small\textnormal{IE}]{\Gamma\vdash_c \phi}{\Gamma\vdash_c \phi\dot{\to}\psi&\Gamma\vdash_c\phi}
	\\[0.3cm]
	\infer[\small\textnormal{CI}]{\Gamma\vdash_c \phi\dot{\land}\psi}{\Gamma\vdash_c \phi& \Gamma\vdash_c \psi}\hspace{3em}
	\infer[{\small\textnormal{CE}_1}]{\Gamma\vdash_c \phi}{\Gamma\vdash_c \phi\dot{\land}\psi}\hspace{3em}
	\infer[{\small\textnormal{CE}_2}]{\Gamma\vdash_c \psi}{\Gamma\vdash_c \phi\dot{\land}\psi}
	\\[0.3cm]
	\infer[{\small\textnormal{DI}_1}]{\Gamma\vdash_c \phi\dot{\lor}\psi}{\Gamma\vdash_c \phi}\hspace{3em}
	\infer[{\small\textnormal{DI}_2}]{\Gamma\vdash_c \phi\dot{\lor}\psi}{\Gamma\vdash_c \psi}\hspace{3em}
	\infer[\small\textnormal{DE}]{\Gamma\vdash_c\theta}{\Gamma\vdash_c\phi\dot{\lor}\psi&\Gamma,\phi\vdash_c \theta&\Gamma,\psi\vdash_c \theta}
	\\[0.3cm]
	\infer[\small\textnormal{AI}]{\Gamma\vdash_c\dot{\forall} \phi}{\up \Gamma\vdash_c \phi}\hspace{3em}
	\infer[\small\textnormal{AE}]{\vphantom{\dot{\forall}}\Gamma\vdash_c \phi[t]}{\Gamma\vdash_c \dot \forall \phi}\hspace{3em}
	\infer[\small\textnormal{EI}]{\Gamma\vdash_c \dot{\exists}\phi}{\Gamma\vdash_c \phi[t]}\hspace{3em}
	\infer[\small\textnormal{EE}]{\vphantom{\dot{\forall}}\Gamma\vdash_c \psi}{\Gamma\vdash_c\dot{\exists}\phi&\up \Gamma,\phi\vdash_c \up \psi}
	\\[0.3cm]
	\hspace{1cm}\infer[\small\textnormal{P}]{\Gamma\vdash_c ((\phi\dot\to\psi)\dot\to\phi)\dot\to\phi}{}
	\end{gather*}
	We write $\vdash_c\phi$ whenever $\phi$ is classically provable from the empty context.
\end{definition}

\setCoqFilename{Gentzen}
\begin{definition}[][sprv]
	The intuitionistic sequent calculus LJT is defined as follows:
	\begin{gather*}
	\hspace{1em}
	\infer[\small\textnormal{A}]{\Gamma \,; \varphi \sq \varphi}{}\hspace{3em}
	\infer[\small\textnormal C]{\Gamma \sq \psi}{\Gamma \,; \varphi \sq \psi \quad \varphi \in \Gamma}\hspace{3em}
	\infer[\small\textnormal{IL}]{\Gamma \,; \varphi \dot\to \psi\sq \theta}{\Gamma \sq \varphi \quad \Gamma \,; \psi \sq \theta}\\[0.3cm]
	\infer[\small\textnormal{IR}]{\Gamma \sq \varphi \dot\to \psi}{\Gamma,\varphi \sq \psi}\hspace{3em}
	\infer[\small\textnormal{AL}]{\Gamma \,; \dot\forall \varphi \sq \psi}{\Gamma \,; \varphi[t]\sq \psi}\hspace{3em}
	\infer[\small\textnormal{AR}]{\Gamma \sq \dot\forall \varphi}{\up{\Gamma} \sq \varphi}\hspace{3em}
	\infer[\small\textnormal E]{\Gamma \sq \varphi}{\Gamma \sq \dot\bot}
	\end{gather*}
\end{definition}

\setCoqFilename{FullSequent}
\begin{definition}[][fprv]
	\label{def:lj}
	The intuitionistic sequent calculus LJ is defined as follows:
	\begin{gather*}
	\infer[\small\textnormal{A}]{\cextend{\Gamma}{\varphi} \fsq \varphi}{}\hspace{3em}
	\infer[\small\textnormal{C}]{\cextend{\Gamma}{\varphi} \fsq \psi}{\cextend{\cextend{\Gamma}{\varphi}}{\varphi} \fsq \psi}\hspace{3em}
	\infer[\small\textnormal{W}]{\cextend{\Gamma}{\varphi} \fsq \psi}{\Gamma \fsq \psi}\hspace{3em}
	\\[0.3cm]
	\infer[\small\textnormal{P}]{\Gamma, \varphi, \psi, \Gamma' \fsq \theta}{\Gamma, \psi, \varphi, \Gamma' \fsq \theta}\hspace{3em}
	\infer[\small\textnormal{E}]{\Gamma \fsq \varphi}{\Gamma \fsq \dot\bot}\hspace{3em}
	\infer[\small\textnormal{IL}]{\cextend{\Gamma}{\varphi \dot\to \psi} \fsq \theta}{\Gamma \fsq \varphi \quad \cextend{\Gamma}{\psi} \fsq \theta}\hspace{3em}
	\\[0.3cm]
	\infer[\small\textnormal{IR}]{\Gamma \fsq \varphi \dot\to \psi}{\cextend{\Gamma}{\varphi} \fsq \psi}\hspace{3em}
	\infer[\small\textnormal{CL}]{\cextend{\Gamma}{\varphi \dot\wedge \psi} \fsq \theta}{\cextend{\cextend{\Gamma}{\varphi}}{\psi} \fsq \theta}\hspace{3em}
	\infer[\small\textnormal{CR}]{\Gamma \fsq \varphi \dot\wedge \psi}{\Gamma \fsq \varphi \quad \Gamma \fsq \psi}\hspace{3em}
	\\[0.3cm]
	\infer[\small\textnormal{DL}]{\cextend{\Gamma}{\varphi \dot\vee \psi} \fsq \theta}{\cextend{\Gamma}{\varphi} \fsq \theta \quad \cextend{\Gamma}{\psi} \fsq \theta}\hspace{3em}
	\infer[{\small\textnormal{DR}_1}]{\Gamma \fsq \varphi \dot\vee \psi}{\Gamma \fsq \varphi}\hspace{3em}
	\infer[{\small\textnormal{DR}_2}]{\Gamma \fsq \varphi \dot\vee \psi}{\Gamma \fsq \psi}\hspace{3em}
	\\[0.3cm]
	\infer[\small\textnormal{AL}]{\cextend{\Gamma}{\dot\forall \varphi} \fsq \psi}{\cextend{\Gamma}{\subst{\varphi}{t}} \fsq \psi}\hspace{2em}
	\infer[\small\textnormal{AR}]{\Gamma \fsq \dot\forall \varphi}{\cshift{\Gamma} \fsq \varphi}\hspace{2em}
	\infer[\small\textnormal{EL}]{\cextend{\Gamma}{\dot\exists \varphi} \fsq \psi}{\cextend{\cshift{\Gamma}}{\varphi} \fsq \shift{\psi}}\hspace{2em}
	\infer[\small\textnormal{ER}]{\Gamma \fsq \dot\exists \varphi}{\Gamma \fsq \subst{\varphi}{t}}\hspace{3em}
	\end{gather*}
\end{definition}

\bibliographystyle{abbrv}
\bibliography{paper}

\end{document}

%% file: macros.tex
\usepackage{multicol}
\usepackage{xspace}
\usepackage{amsmath}
\usepackage{amssymb}
\usepackage{stmaryrd}
\usepackage{proof}

\newtheorem{fact}[lemma]{Fact}

\newcommand{\MBB}[1]{\ensuremath{\mathbb{#1}}\xspace}
\newcommand{\MCL}[1]{\ensuremath{\mathcal{#1}}\xspace}
\newcommand{\MSF}[1]{\ensuremath{\mathsf{#1}}\xspace}
\newcommand{\MFR}[1]{\ensuremath{\mathfrak{#1}}\xspace}
\newcommand{\Prop}{\MBB{P}}

\newcommand{\Nat}{\MBB{N}}
\newcommand{\Bool}{\MBB{B}}
\newcommand{\Term}{\MBB{T}}
\newcommand{\Form}{\MBB{F}}
\newcommand{\sForm}{\Form^*}

\newcommand{\btrue}{\mathsf{tt}}
\newcommand{\bfalse}{\mathsf{ff}}

\newcommand{\natS}{\mathsf{S}}
\newcommand{\some}[1]{\ulcorner#1\urcorner}

\newcommand{\Opt}{\MCL{O}}
\newcommand{\List}{\MCL{L}}
\newcommand{\app}{+\hspace{-6pt}+\,}

\newcommand{\Funcs}{\MCL{F}_\Sigma}
\newcommand{\Preds}{\MCL{P}_\Sigma}

\renewcommand{\H}{\MCL H}
\newcommand{\HO}{\ensuremath{\H_c}\xspace}
\newcommand{\LB}{\MFR{L}}
\newcommand{\UB}{\MFR{U}}
\newcommand{\down}{\!\!\Downarrow\,}
\newcommand{\up}{\uparrow\!\!}
\newcommand{\MP}{\MSF{MP}}
\newcommand{\EM}{\MSF{EM}}
\newcommand{\TT}{\MCL{T}}
\newcommand{\MM}{\MCL{M}}
\newcommand{\KK}{\MCL{K}}
\newcommand{\UU}{\MCL{U}}
\newcommand{\CC}{\MCL{C}}
\newcommand{\WW}{\MCL{W}}

\renewcommand{\phi}{\varphi}
\newcommand{\heval}[1]{[\![ #1 ]\!]}
\newcommand{\sq}{\hspace{-0.2em}\Rightarrow\hspace{-0.2em}}

\newcommand{\binop}{\boxdot}

\newcommand{\id}{\MSF{id}}

\newcommand{\attack}[3]{#3 \,|\, #2 \rhd #1}
\newcommand{\attackN}[2]{#2 \rhd #1}
\newcommand{\attacks}{\mathcal{A}}
\newcommand{\defenses}[1]{\mathcal{D}_{#1}}
\newcommand{\epmove}[2]{#1 \leadsto^E_p #2}
\newcommand{\eomove}[2]{#1 \leadsto^E_o #2}
\newcommand{\dpmove}[2]{#1 \leadsto^D_p #2}
\newcommand{\domove}[2]{#1 \leadsto^D_o #2}
\newcommand{\spmove}[2]{#1 \leadsto^S_p #2}
\newcommand{\somove}[2]{#1 \leadsto^S_o #2}
\newcommand{\eWin}[1]{\text{Win}^E\,#1}
\newcommand{\dWin}[1]{\text{Win}^D\,#1}
\newcommand{\sWin}[1]{\text{Win}^S\,#1}
\newcommand{\just}[2]{\text{justified }#1\,#2}
\newcommand{\subsq}[1]{\Rightarrow_{#1}}
\newcommand{\LJD}{\subsq{D}}
\newcommand{\fsubsq}[1]{\Rightarrow_{#1}}

\newcommand{\fsq}{\fsubsq{J}}

\newcommand{\sstate}{\List(\Form) \times \List(\Form) \times \List(\attacks \times \attacks)}

\newcommand{\extend}[2]{#1, #2}
\newcommand{\cextend}[2]{\extend{#1}{#2}}

\newcommand{\shift}[1]{\uparrow \hspace{-0.2em} #1}
\newcommand{\cshift}[1]{\uparrow \hspace{-0.2em} #1}
\newcommand{\subst}[2]{#1[#2]}
\newcommand{\dstate}{\List(\Form) \times \List(\attacks) \times \List(\Form) \times \List(\attacks)}

\newcommand{\siglift}[1]{\Uparrow \hspace{-0.2em} #1}
\newcommand{\sigdrop}[1]{\Downarrow \hspace{-0.2em} #1}
\newcommand{\sigdropn}[2]{\Downarrow^{#1} \hspace{-0.2em} #2}

\renewcommand{\L}{\mathsf{L}}
\newcommand{\MPL}{\mathsf{MP}_\L}
\newcommand{\eva}{\mathcal{E}}